%% file: eigenvalue_sampling.tex
\documentclass[11pt]{article}

\usepackage{setspace}
\usepackage{algorithm}
\usepackage{algpseudocode}

\usepackage[style=alphabetic]{biblatex}
\date{}

\input{styling.sty}
\newcommand{\outlyingthresh}{L}
\title{Tight Sampling Bounds for Eigenvalue Approximation}

\begin{document}
\title{\Large Tight Sampling Bounds for Eigenvalue Approximation 
\thanks{The authors thank support from a Simons Investigator Award, NSF CCF-2335412, and a Google Faculty Award.}
}
\author{William Swartworth \\Carnegie Mellon University 
\and David P. Woodruff \\Carnegie Mellon University
}
\date{}

\maketitle

\begin{abstract}
\input{abstract}
\end{abstract}

\newpage
\tableofcontents
\newpage

\section{Introduction}
\input{intro}
 
\input{technical_sections}

\printbibliography

\end{document}

%% file: abstract.tex
We consider the problem of estimating the spectrum of a symmetric bounded entry (not necessarily PSD) matrix via entrywise sampling.  This problem was introduced by [Bhattacharjee, Dexter, Drineas, Musco, Ray '22], where it was shown that one can obtain an $\epsilon n$ additive approximation to all eigenvalues of $A$ by sampling a principal submatrix of dimension $\frac{\text{poly}(\log n)}{\epsilon^3}$.  We improve their analysis by showing that it suffices to sample a principal submatrix of dimension $\tilde{O}(\frac{1}{\epsilon^2})$ (with no dependence on $n$).  This matches known lower bounds and therefore resolves the sample complexity of this problem up to $\log\frac{1}{\epsilon}$ factors.  Using similar techniques, we give a tight $\tilde{O}(\frac{1}{\epsilon^2})$ bound for obtaining an additive $\epsilon\|A\|_F$ approximation to the spectrum of $A$ via squared row-norm sampling, improving on the previous best $\tilde{O}(\frac{1}{\epsilon^{8}})$ bound.
We also address the problem of approximating the top eigenvector for a bounded entry, PSD matrix $A.$   In particular, we show that sampling $O(\frac{1}{\epsilon})$ columns of $A$ suffices to produce a unit vector $u$ with $u^T A u \geq \lambda_1(A) - \epsilon n$.  This matches what one could achieve via the sampling bound of [Musco, Musco'17] for the special case of approximating the top eigenvector, but does not require adaptivity. 

As additional applications, we observe that our sampling results can be used to design a faster eigenvalue estimation sketch for dense matrices resolving a question of [Swartworth, Woodruff'23], and can also be combined with [Musco, Musco'17] to achieve $O(1/\epsilon^3)$ (adaptive) sample complexity for approximating the spectrum of a bounded entry PSD matrix to $\epsilon n$ additive error.

%% file: intro.tex
Computing the spectrum of a matrix is a fundamental problem with many applications.  There are well-known high-precision algorithms that run in polynomial time \cite{francis1962qr, golub2000eigenvalue}, although any such algorithm is necessarily at least linear time in the input size.  As data grows larger, even linear algorithms can be prohibitive. This has motivated a flurry of activity studying sublinear time estimation of problems in numerical linear algebra, for instance for low-rank approximation \cite{musco2017recursive, musco2017sublinear, bakshi2018sublinear, bakshi2020robust}, kernel density estimation \cite{charikar2017hashing, siminelakis2019rehashing, charikar2020kernel}, testing positive-semidefiniteness \cite{bakshi2020testing}, and matrix sparsification \cite{bhattacharjee2023universal, drineas2011note}.

For eigenvalue estimation, variants of the power method have long been known to give good approximations to the top eigenvalues and eigenvectors of $A\in\R^{n\times n}$ while revealing sublinear information about $A$, i.e., using $o(n)$ matrix-vector queries \cite{rokhlin2010randomized, musco2015randomized}. However it was only recently asked in \cite{bhattacharjee2024sublinear} whether there are spectral approximation algorithms for symmetric, but non-PSD matrices that run in sublinear time in the entry query model.  This is perhaps the most natural model if one imagines having an extremely large matrix saved on disk for example. This may be in the form of a graph for instance, where one could be interested in obtaining spectral information about its Laplacian or adjacency matrix.  One could also imagine having a large collection of data points with some kernel function that can be computed for pairs of points.  Obtaining a rough spectral summary of the associated kernel matrix is a natural step for data analysis, for instance, to spot low-rank structure in the data.  If data points are large or expensive to collect, or if kernel evaluation is expensive, it is natural to aim for minimizing entry queries to the kernel matrix.

Of course, it is not reasonable to ask for sublinear time spectral approximation algorithms, without some additional assumptions.  For instance, our matrix $A$ could contain all zeros but with a single large entry at indices $(i,j)$ and $(j,i).$  Given only the ability to query entries, and no additional information, even distinguishing $A$ from the all zeros matrix would take $\Omega(n^2)$ queries.

We consider two assumptions that allow for improved guarantees.  The first is an assumption on the structure of $A$ called the \textit{bounded entry model}, which assumes that $A$ has entries bounded by $1$ in magnitude. This condition was introduced in \cite{balcan2019testing} and studied further by \cite{bakshi2020testing} who showed that it was sufficient in order to test for positive semi-definiteness with sublinear entry queries.  Motivated by this result, \cite{bhattacharjee2024sublinear} showed that all eigenvalues of a symmetric bounded-entry matrix can be approximated using a sublinear number of queries, simply by sampling a $\poly(\frac{\log n}{\eps})$ sized submatrix.

Another way of getting sublinear sample complexity is to give the sampler additional power.
In our case, as in \cite{bhattacharjee2024sublinear}, we consider having access to a sampler that can produce a row index with probability proportional to its squared row norm.
Such samplers have been increasingly studied under the guise of ``quantum-inspired" machine-learning algorithms \cite{tang2019quantum, chepurko2022quantum, gilyen2018quantum, gilyen2022improved}.  Such samplers are practical to maintain when $A$ is stored entrywise.  For example by using an appropriate data structure, they can be built in $\nnz(A)$ time, admit $O(\log n)$ time sampling, and can handle entry updates in $O(\log n)$ time.


\cite{bhattacharjee2024sublinear} showed that given a real symmetric matrix $A$ with entries bounded by $1$, one can sample a principal submatrix of $A$ of dimensions $O(\frac{\poly\log n}{\eps^3}) \times O(\frac{\poly\log n}{\eps^3})$ and then output an additive $\eps n$ approximation to the entire spectrum, i.e., to all eigenvalues of $A.$  On the other hand, the best lower bound states that a principal-submatrix algorithm must sample at least $O(\frac{1}{\eps^4})$ entries of $A.$ 

There are two ways that one could hope to improve the sampling bound of \cite{bhattacharjee2024sublinear}.  First one could hope to improve the $\eps$ dependence in the dimension from $O(1/\eps^3)$ to $O(1/\eps^2).$  Several prior results suggested that this might be possible.  For example \cite{bhattacharjee2024sublinear} showed that their $\eps$ dependence could be improved both when $A$ is PSD, and when the spectrum of $A$ is flat. Concurrently \cite{swartworth2023optimal} showed that one can obtain an $\eps \norm{A}{F}$ additive approximation to the spectrum of $A$ by using a so-called bilinear sketch of $A$ of dimensions $O(1/\eps^2) \times O(1/\eps^2).$  Such a sketch would give $\eps n$ additive error for approximating all eigenvalues when $A$ has bounded entries.  Unfortunately this sketch is Gaussian, and it seems difficult to directly adapt its analysis to obtain a sampling bound instead. 



In this paper we close the gap between sketching and sampling for bounded entry matrices by showing that uniformly sampling an $\tilde{O}(1/\eps^2) \times \tilde{O}(1/\eps^2)$ principal submatrix of $A$ suffices to approximate all eigenvalues of $A$ up to $\eps n$ additive error, even when $A$ is not necessarily PSD. In addition to obtaining an optimal $\eps$-dependence, we note that our uniform sampling bound contains no dependence on $n.$

We also address the squared-row norm sampling model. Here we improve the analysis of \cite{bhattacharjee2024sublinear} to show that it suffices to query a principal submatrix of size $\tilde{O}(\frac{\poly \log n}{\eps^2})$, compared to the $\tilde{O}(\frac{\poly \log n }{\eps^8})$ dimensional principal submatrix required by \cite{bhattacharjee2024sublinear}.

\paragraph{Our approximation model.} We note that all of the guarantees considered in our work and prior work focus on additive approximations to the spectrum.  Ideally, one might like to aim for a relative error guarantee.  However as pointed out, by \cite{bhattacharjee2024sublinear} for example, this is not possible for entry queries.  Such an algorithm would be able to distinguish the $0$ matrix from a matrix with a single off-diagonal pair of nonzero entries, which clearly requires $\Omega(n^2)$ samples. Indeed even with squared row-norm sampling, relative error is still too much to hope for. In fact, even for sketches, approximating the top eigenvalue to within a constant factor requires $\Omega(n^2)$ sketching dimension in general \cite{li2016tight, woodruff2014sketching}.  One can always turn such a sketching lower bound into a sampling bound, even allowing for row-norm sampling - simply conjugate by a random orthogonal matrix to flatten all rows.  Then squared row sampling is effectively uniform, and so a sampling algorithm could be used to construct a sketch of the same dimensions.  These existing lower bounds are why we (as well as prior work) choose to focus on an additive approximation guarantee.

\section{Our Results}

\input{results_table}

As discussed above, we are interested in the same type of guarantee as considered in \cite{bhattacharjee2024sublinear} and \cite{swartworth2023optimal}.
\begin{Definition}
    Let $A \in \R^{n\times n}$ be a symmetric matrix.  We say that a sequence $\hat{\lambda}_1 \geq \hat{\lambda}_2\geq \ldots \geq \hat{\lambda}_n$ is an additive $\alpha$-approximation to the spectrum of $A$ if $\abs{\hat{\lambda}_i(A) - \lambda_i(A)} \leq \alpha$ for all $i\in [n].$
\end{Definition}
We consider the sampling algorithms introduced by \cite{bhattacharjee2024sublinear}, with some technical modifications. Our novelty is an improved analysis which yields optimal bounds up to logarithmic factors.  Our techniques apply both to uniform sampling for bounded entry matrices, as well as to squared row norm sampling for arbitrary matrices.

\paragraph{Uniform sampling.} For uniform sampling, we show the following result.
\input{uniform_sampling}
\begin{theorem}
Let $A$ be a symmetric matrix with all entries bounded by $1$.  Then Algorithm~\ref{alg:uniform_sampling} with $s\geq c \frac{1}{\eps^2}\log^2\frac{1}{\eps}$ (where $c$ is an absolute constant) outputs an additive $\eps n$ approximation to the spectrum of $A$, with at least $2/3$ probability.
\end{theorem}

The sampling algorithm used here is perhaps as natural as possible.  We simply sample a random principal submatrix of expected size $s\times s$ and appropriately rescale.  We then output the eigenvalues of the sampled submatrix, along with roughly $n-s$ additional $0$'s to account for the remaining spectrum of $A.$ Our sampling bound is an improvement over \cite{bhattacharjee2024sublinear} which requires sampling a principal submatrix of dimension $\frac{c}{\eps^3} \log\frac{1}{\eps}\log^3 n$ to obtain the same guarantee. We note that \cite{bhattacharjee2024sublinear} was able to push their entry sample complexity down to $O(1/\eps^5)$ by using a matrix sparsification result to subsample their principal submatrix.  However the resulting sample is not a principal submatrix, and it is unclear how to decrease the sample complexity further.

\cite{bhattacharjee2024sublinear} also asked whether one can obtain an additive $\eps n$ guarantee to the spectrum of $A$ with sample size independent of $A.$  Our result shows that this is indeed possible.  This may seem surprising since \cite{bhattacharjee2024sublinear} uses a bound due to Tropp \cite{rudelson2007sampling, tropp2008norms}, to control the contribution from the small-magnitude eigenvalues. This bound contains a $\log n$ term that appears difficult to remove, and we use the same bound on the small-magnitude eigenvalues. Nonetheless, we give a very simple bootstrapping argument to show that this $\log n$ dependence is non-essential.  Specifically, we show that it is possible to replace the $\log n$ with a $\log\frac{1}{\eps}$, thus removing the $n$ dependence, at the cost of only additional $\log\frac{1}{\eps}$ factors.

\paragraph{Squared row-norm sampling.} For squared row-norm sampling, we use our analysis to again give nearly tight bounds.
\input{row_norm_sampling}
\begin{theorem}
Let $A$ be an arbitrary symmetric matrix.  Algorithm~\ref{alg:row_norm_sampling} with $s\geq \frac{1}{\eps^2}\poly\log(n/\eps)$ outputs an additive $\eps\norm{A}{F}$ additive approximation to the spectrum of $A$ with at least $2/3$ probability.
\end{theorem}

This is a substantial improvement over the sampling bound in \cite{bhattacharjee2024sublinear}, which requires sampling a principal submatrix of dimension
$\frac{c\log^{10} n}{\eps^8}$ to match our bound.

Again the most challenging part of our analysis is in controlling the outlying eigenvalues.  Here \cite{bhattacharjee2024sublinear} again requires roughly $1/\eps^3$ samples, whereas we require $1/\eps^2$ samples by essentially the same argument as for the case of uniform sampling.  However on its own, this is enough to give improved bounds as the bottleneck in their argument is in controlling the middle eigenvalues.  The authors do not attempt to give tight $\eps$ bounds for this part of their analysis, and note that their bound is likely not tight.  However we observe that their analysis here is surprisingly close to optimal.  By optimizing their variance argument, we observe that their approach is sufficient to control the middle eigenvalues with optimal $\eps$ dependence.

We note that one cannot hope to do better than $O(1/\eps^4)$ queries for non-adaptive entry sampling algorithms.  Indeed a lower bound of \cite{li2016tight} implies that learning the top singular value of a (non-symmetric) matrix $G \in \R^{n\times n}$ to $\eps \norm{G}{F}$ additive error requires $\Omega(1/\eps^4)$ sketching dimension.  One can symmetrize to $G' = [0, G; G^T, 0]$ to show that learning the top eigenvalue of $G'$ to $\eps \norm{G'}{F}$ additive error requires $\Omega(1/\eps^4)$ sketching dimension in general.  The hard instance $G$ is Gaussian, so its entries are bounded by $\eps^2 \norm{G}{F}^2$ up to log factors implying an $\tilde{\Omega}(1/\eps^4)$ lower bound in the bounded entry model, even for sketches.  The rows of $G$ have equal norms up to constant factors say, and so this implies that $\tilde{\Omega}(1/\eps^4)$ queries are needed even given access to a row-norm sampling oracle that is accurate up to constant factors (which is all that our sampling algorithm requires).

Finally, note that while our results are not stated as a high-probability guarantee, it is easy to improve them.  Simply run the algorithm $O(\log\frac{1}{\delta})$ times and take the median estimate for each eigenvalue $\lambda_i.$

\paragraph{Top eigenvector estimation.}
While our main results focus on eigenvalue estimation, in many situations, one is interested not only in the top eigenvalues but also the associated top eigenvectors.  As a step in this direction, we show that for bounded entry PSD matrices, one can obtain an $\eps n$-approximate top eigenvector by sampling just $O(1/\eps)$ columns of $A.$

The same guarantee could be achieved by \cite{musco2017recursive}, by using $\tilde{O}(n/\eps)$ \textit{adaptive} entry samples to obtain an additive spectral approximation to $A$. Interestingly, in the bounded entry case, we show that no adaptivity is needed.

\input{top_eigvect_algorithm}
\begin{theorem}
Let $A\in\R^{n\times n}$ be a symmetric PSD with $\norm{A}{\infty}\leq 1$.  For $p = \frac{c}{\eps n}$, with $3/4$ probability, Algorithm~\ref{alg:top_eigenvector} returns a unit vector $u\in\R^n$ satisfying
\[
u^T A u \geq \lambda_1 - \eps n,
\]
where $\lambda_1$ is the top eigenvalue of $A.$
\end{theorem}

\paragraph{Faster sketching.}  \cite{swartworth2023optimal} studied the problem of sketching the eigenvalues of $A$ to $\eps\norm{A}{F}$ additive error, and showed that this can be accomplished with a Gaussian sketch of dimension $O(1/\eps^4).$   Their sketch takes roughly $\frac{1}{\eps^2}n^2$ time to apply, so if $\eps$ were say $n^{-1/4},$ this would be substantially slower than linear in the size of $A$.  It was left open whether faster sketching is possible here when $A$ is dense.  We observe that our row-norm sampling analysis implies the existence of a sketch that runs in $\tilde{O}(n^2)$ time with sketching dimension $\tilde{O}(1/\eps^4).$

\subsection{Our Techniques}

\paragraph{Bounded Entry Matrices.}
We use the same algorithm that was introduced in \cite{bhattacharjee2024sublinear}, however we show how to reduce the sample size without sacrificing in terms of the approximation guarantee.  At a high level we follow a similar approach of splitting the matrix $A$ as $A_o + A_m$ into its ``outer" and ``middle" eigenvalues.  The matrix $A_o$ will zero out all eigenvalues of $A$ with magnitude smaller than $\eps n.$  We will think of $A_o$ as containing all the eigenvalues that we need to approximate, and $A_m$ as adding some additional noise whose contribution we must bound.  For a bounded entry matrix $A$, note that $\norm{A}{F}^2\leq n^2$, so that there can be at most $1/\eps^2$ nonzero eigenvalues in $A_o.$  Thus it is at least reasonable to hope that an $O(1/\eps^2)$ sized principal submatrix is sufficient.
  
If $S$ is the (appropriately rescaled) sampling matrix, then we can write $S A S^T = S A_o S^T + S A_m S^T.$  Note that $SAS^T$ is simply our principal submatrix sample of $A.$

Given this decomposition, the analysis naturally decomposes into two parts: showing that the eigenvalues of $SA_oS^T$ concentrate appropriately around the large eigenvalues, and showing that the perturbation from $SA_m S^T$ is small.  Prior work \cite{bhattacharjee2024sublinear} handles the middle eigenvalues nearly optimally (although we give a small improvement here that we describe later).  So our main technical innovation is to show that the eigenvalues of $S A_o S^T$ concentrate to within $\eps n$ of the eigenvalues of $A_o$ when $S$ samples only $\tilde{O}(1/\eps^2)$ rows.  

To see the main ideas, write the spectral decomposition of $A_o$ as $A_o = V_o \Lambda_o V_o^T$, where $V_o$ contains the top eigenvectors of $A_o$ (where ``top" eigenvectors  have eigenvalues at least $\eps n$ in magnitude).  Consider the extreme case where $A_o$ has a flat spectrum -- say $k$ copies of the eigenvalue $\lambda$ so that $A_o$
 is the scaled projection $\lambda V_o V_o^T.$  After applying $S$, we are left with the matrix $\lambda SV_o V_o^T S^T$ whose spectrum coincides with $\lambda V_o^T S^T S V_o.$

 This suggests showing that $S$ is a subspace embedding for $V_o$.  Indeed, if $S$ distorts by at most $(1\pm \alpha)$ on $V_o$, then for all $x$ we would have 
 \[
 x^T (\lambda V_o^T S^T S V_o) x 
 = \lambda \norm{SV_o x}{2}^2
 = (1 \pm \alpha) \lambda,
 \]
 so that the quadratic form for $S V_o S^T$ is always at most $\lambda + \alpha\lambda$ on unit vectors $x$, and is at least $\lambda - \alpha\lambda$ on a $k$-dimensional subspace.  By Courant-Fischer this would be sufficient to show concentration of the $k$ eigenvalues of $SV_oS^T$ to within $\eps n$ of those of $V_o$ as long as $\alpha \leq \frac{\eps n}{\lambda}.$ 

In the worst case, $\lambda$ could be as large as $n$ if $A$ were the all-ones matrix for example.  Also since we allow negative eigenvalues it is possible to have up to $1/\eps^2$ eigenvalues of magnitude $\eps n.$  This suggests that we choose $S$ so that it yields a $(1\pm \eps)$-distortion subspace embedding over a $1/\eps^2$ dimensional space.  Unfortunately, this would require $S$ to have at least $1/\eps^3$ rows, which would only allow us to match the sampling bound of \cite{bhattacharjee2024sublinear}.  In order to do better we observe similarly to \cite{swartworth2023optimal}, that large eigenvalues cannot occur many times in the spectrum.  Since $A$ has bounded entries, its Frobenius norm is at most $n$, which means that there are at most $n^2/\lambda^2$ eigenvalues of magnitude at least $\lambda.$  So we could hope to only require a $1\pm (\eps n/\lambda)$ distortion subspace embedding over a space of dimension $n^2/\lambda^2$, which can plausibly be achieved by an $S$ with roughly $(n^2/\lambda^2)\cdot(\lambda/(\eps n))^2 = 1/\eps^2$ rows.

Indeed, by using the ``incoherence bound" of \cite{bhattacharjee2024sublinear}, we observe that the leverage scores of $V_o$ are sufficiently small so that uniform sampling approximates leverage score sampling.  This allows to argue that sampling roughly $1/\eps^2$ rows is sufficient to obtain a $1 \pm (\eps n / \lambda)$ distortion subspace embedding on $V_{\geq \lambda}$, the eigenvectors with associated eigenvalue at least $\lambda$ in magnitude.  It is possible however that the spectrum of $A$ is not completely flat, and contains a range of outlying eigenvalues between $\eps n$ and $n.$  To handle this, we ask for $S$ to satisfy the following deterministic condition:

\begin{center}
$S$ is a $1 \pm (\eps n / \lambda)$ distortion subspace embedding on $V_{\geq \lambda}$ \textit{for all $\lambda \geq \eps n$}.
\end{center}

Perhaps surprisingly, we show that this simple condition is enough to approximate all eigenvalues of $A_o$ to within $\eps n$ additive error.

While the above paragraphs motivate this condition, there are a few technical challenges.  The first is handling negative eigenvalues. If $A$ were PSD, then the subspace embedding argument given above would essentially be sufficient to show that $S A S^T$ is at least $\lambda_k - \eps n$\  on a subspace of dimension $k$, which would yield the right lower bound on the $k$th eigenvalue of $S A S^T.$  However we have to worry that the negative eigenvalues of $A$ could bring down this  eigenvalue estimate. We indeed find a $k$-dimensional subspace on which $SAS^T$ is at least $\lambda_k - \eps n$ which is mostly aligned with the top $k$ eigenvectors of $A$. However we choose it carefully so as to avoid the large negative eigenvectors of $A.$  This is similar to the approach suggested in \cite{swartworth2023optimal}, however their argument crucially relied on $S$ being Gaussian so that the positive and negative eigen-spans of $SAS^T$ are nearly orthogonal.  Our argument shows that this is not necessary; the simple deterministic guarantee stated above suffices.

The other main challenge is upper bounding the top eigenvalue of $S A S^T$ (which turns out to be sufficient for upper bounding all eigenvalues).  \cite{swartworth2023optimal} gave a bound on the top eigenvalue of a Gaussian matrix with non-identity covariance.  Unfortunately in the sampling setting, we cannot rely on Gaussianity.  Instead, we show that that the same deterministic subspace embedding condition given above that was natural for lower-bounding eigenvalues also suffices for obtaining the necessary upper bounds on the eigenvalues.  The proof works by first partitioning eigenvalues into level sets.  It is simple to prove the desired upper bound on the top eigenvalue of $SAS^T$ if $A$ were restricted to be just the top level set. Then by carefully keeping track of the interaction between each pair of level sets, we are able to show that the additional level sets do not add much mass to the top eigenvalue.

Finally we remove the extra $\log n$ arising from the bound on the middle eigenvalues \cite{bhattacharjee2024sublinear}.  The idea is simply to observe that after applying the uniform sampling algorithm once, we are again left with a roughly $s\times s$ matrix, whose spectrum we would like to approximate to additive $O(\eps n)$ error.  After rescaling, the entries of this matrix are bounded by $n/s,$ and so we can now apply the uniform sampling algorithm again to obtain an $O(\eps\cdot \frac{n}{s}\cdot s) = O(\eps n)$ additive approximation to the spectrum of the submatrix, this time by sampling $O(\frac{1}{\eps^2}\poly\log s)$ rows.  Note that $\log n$ is now replaced with $\log s$ which is potentially much smaller.  If needed, this procedure can be repeated until the dependence on $\log n$ is removed.

\paragraph{Squared row-norm sampling.} As with uniform sampling, our algorithm is based off the algorithm of \cite{bhattacharjee2024sublinear}, although with a tighter analysis.  Their result requires that one samples a principal submatrix of dimension roughly $s\times s$ with $s = O(\frac{1}{\eps^8}\poly\log n)$ in order to obtain an additive $\eps \norm{A}{F}$ approximation to the spectrum of $A.$  Notably, their algorithm does not simply return the spectrum of the sampled submatrix but instead judiciously zeros out certain entries in order to reduce variance. To see why this is necessary, it is instructive to consider the case where $A$ is the identity matrix.  If one samples a $1/\eps^2$ sized principal submatrix and rescales it by $n \eps^2$, then the resulting matrix will have eigenvalues of size $\eps^2 n.$ This would be fine if one wanted to achieve $\eps n$ additive error.  However, for $\eps \norm{A}{F} = \eps \sqrt{n}$ additive error, this is unacceptable.  Fortunately, \cite{bhattacharjee2024sublinear} gives an entry zeroing procedure which zeros out most diagonal entries as well as entries in particularly sparse rows and columns.  We are able to employ their zeroing procedure essentially as a black-box.

We do not  modify their zeroing procedure, but instead improve their bounds for both the outer and middle eigenvalues of the sampled submatrix.  We use the same decomposition $SAS^T = SA_o S^T + S A_m S^T$ as for the uniform sampling analysis discussed above, where $S$ now samples each row $i$ with probability $p_i$ proportional to its squared row norm, and rescales by $1/\sqrt{p_i}.$  Our task is two-fold.  We first show that the eigenvalues of $S A_o S^T$ concentrate within $\eps \norm{A}{F}$ to the eigenvalues of $A_o.$ Then we bound the operator norm of $SA_m S^T$ by $\eps \norm{A}{F}$, and so Weyl's inequality shows that the eigenvalues of $S A S^T$  are within $O(\eps n)$ of the eigenvalues of $A.$

Bounding the eigenvalues of $S A_o S^T$ is similar to the uniform sampling case.  Squared row-norm sampling provides an approximation to leverage score sampling, which is sufficient to obtain the subspace embedding guarantees that we require.

Bounding the operator norm of $SA_m S^T$ is what required \cite{bhattacharjee2024sublinear} to have a $1/\eps^8$ dependence.  It turns out that a fairly simple technical improvement to their argument improves this to the optimal $1/\eps^2$ dependence, at least when $p_i := \frac{s\norm{A_i}{2}^2}{\norm{A}{F}^2} \leq 1,$ so that the sampling probabilities for each row are all at most $1.$  Their argument contains a technical fix for the situation where $p_i > 1,$ however we do not need this.  Instead we can duplicate the rows with large norm, say $N$ times, while scaling them down by a factor of $1/\sqrt{N}.$  This does not change the spectrum, but reduces to the situation where all $p_i$'s are at most $1.$  The advantage of this is that rows of $A$ with large norm can now be sampled multiple times if the value of $p_i$ dictates that this should occur.  We note that one could likely also sample $s$ rows i.i.d. from the squared-row-norm distribution to achieve a similar result.

\paragraph{Top eigenvector estimation.}  For PSD $A\in\R^{n\times n}$ with entries bounded by $1$, we are interested in producing a vector $u$ that nearly maximizes the Rayleigh quotient $\frac{u^T A u}{u^T u}.$  In the same spirit as our eigenvalue estimation results, we are willing to tolerate $\eps n$ additive error, so we would like to find $u$ that satisfies $\frac{u^T A u}{u^T u} \geq \lambda_1 - \eps n.$

Perhaps the simplest attempt would be to choose a column sampling matrix $S$ and form $S^T A S.$  We now have access to the Rayleigh quotient on the image of $S$, and so we could choose $u\in \text{Im}(S)$ to maximize $\frac{u^T A u}{u^T u}.$  Unfortunately it is clear that this doesn not work.  For example consider the situation where $A$ is the all ones matrix and $S$ samples $k$ columns. Then $\lambda_1 = n,$ but $u$ is supported on $k$ coordinates which means that the Rayleigh quotient is at most $k.$  Even to obtain $n/2$ additive error, we would need to set $k = \Omega(n)$, meaning that we would sample a constant fraction of $A$'s entries.

Searching for $u \in \text{Im}(S)$ is not good enough.   A better approach would be to optimize the Rayleigh quotient over $\text{Im}(AS).$  Indeed if we could implement this approach it would work.  The issue is that computing the Rayleigh quotient for such $U$ would require us to calculate expressions of the form 
\[
\frac{(ASx)^T A (AS x)}{(ASx)^T (ASx)}
= \frac{x^T S^T A^3 S x}{x^T S^T A^2 S x}.
\]
Unfortunately $AS$ does not give us enough information to calculate the numerator. However it does give us enough information to calculate the denominator.  This suggests that we should instead optimize $u$ over $\text{Im}(A^{1/2}S)$. Indeed this only requires us to compute quantities of the form
\[
\frac{x^T S^T A^2 S x}{x^T S^T A S x},
\]
which we can do if we have $AS.$  We show that this works.  That is, for some $x$, the generalized Rayleigh quotient above is at least $\lambda_1 - \eps n.$  Thus if $x$ is optimal, then we can take $A^{1/2} S x$ as our approximate eigenvector. Given $x$, we do not have a direct way to compute $A^{1/2}Sx.$  However we can compute $ASx$.  This can be thought of as applying an additional half iteration of power method to $A^{1/2}Sx$, and so as we observe, the Rayleigh quotation for $ASx$ will be at least as large as for $A^{1/2}Sx.$  Thus our algorithm will ultimately return $ASx.$

\paragraph{Fast eigenvalue sketching.}
Previous work, \cite{swartworth2023optimal} considered the problem of sketching the eigenvalues of a matrix $A$ up to $\eps\norm{A}{F}$ additive error.  While this sketch achieved the optimal $O(\frac{1}{\eps^4})$ sketching dimension, it requires roughly $\frac{n^2}{\eps^2}$ time to apply when $A$ is a dense $n\times n$ matrix.  We can use our sampling analysis to obtain a faster sketch.

The idea is to rotate by a random orthogonal matrix that supports fast matrix multiplication.  Specifically we take a matrix $U$ which is a Hadamard matrix composed with random sign flips and compute $U^T A U.$ This flattens the row-norms of $A$ so that one can simply take the sketch to be a random principle submatrix of $U^T A U$.  Then applying our analysis for squared row-norm sampling shows that the sketch can be used to obtain an additive $\eps\norm{A}{F}$ approximation to the spectrum of $A.$

\subsection{Additional Related Work}

\paragraph{Spectral density estimation.}
Recently, spectral density estimation was studied by \cite{braverman2022linear}, for normalized graph Laplacians and adjacency matrices.  This result measures error with respect to Wasserstein distance for the (normalized) spectral histogram, and queries $\Omega(n/\poly(\eps))$ entries of $A$, so is not directly comparable to our setting.  A related work \cite{cohen2018approximating}, achieves the same type of Wasserstein guarantee but with $\exp(1/\eps)$ queries and access to a random walk on the underlying graph.  Since we aim for $\poly(1/\eps)$ queries this is again not directly comparable to our setting.

\paragraph{Matrix sparsification.}
\cite{bhattacharjee2023universal} studies the problem of deterministically constructing a sparsifier of bounded entry matrix $A$ by using a sublinear number of queries.  Their approximation result would yield an additive $\eps n$ approximation to all eigenvalues but requires $\Omega(n/\eps^2)$ entry queries.  Similarly, for bounded entry PSD matrices they could achieve the same $\Omega(n/\eps^2)$ bound for approximating the top eigenvector of a PSD matrix, but this is worse than what we require for sampling by a factor of $1/\eps.$

\paragraph{Low rank approximation.}
A line of work  \cite{musco2017sublinear, bakshi2018sublinear, bakshi2020robust} considered the problem of constructing low-rank approximations to PSD matrices $A$ using a sublinear number of queries. Notably \cite{bakshi2020robust} gave an algorithm with optimal $O(kn/\eps)$ query complexity for obtaining a PSD rank $k$ approximation $\hat{A}$ to $A$ satisfying
\[
\norm{A - \hat{A}}{F} \leq (1 + \eps) \norm{A - A_{\langle k \rangle}}{F},
\]
where $A_{\langle k \rangle}$ is the optimal rank $k$ approximation of $A.$  For the case $k=1$, this is related to the top eigenvector approximation problem that we consider.  In a similar spirit, \cite{musco2017recursive} gives an additive spectral approximation to $A$ which also implies our top eigenvector bound. However notably these algorithms are all adaptive, whereas we are interested in algorithms that sample non-adaptively. Thus the techniques involved are quite different from what we consider.  Additionally these algorithms all require at least an extra $O(\log\frac{1}{\eps})$ overhead compared to our result.

\paragraph{Sketching.}  In a somewhat different setting to sampling, eigenvalue sketching was previously studied by \cite{swartworth2023optimal} where an optimal $O(1/\eps^4)$ sketching dimension was given for obtaining $\eps\norm{A}{F}$ additive error.  Previously \cite{andoni2013eigenvalues} gave a sketch for the top eigenvalues.  By rearranging their bounds, this implies an $O(1/\eps^6)$ for sketching all eigenvalues to $\eps\norm{A}{F}$ additive error; however this only applies to PSD matrices.  Finally, we note the question of sketching the operator norm was considered in \cite{li2016tight}.  These results imply an $\Omega(n^2)$ lower bound for sketching the top eigenvalue to relative error, which justifies the focus on additive error for sketching and sampling algorithms.

%% file: results_table.tex
\begin{table}
\label{table_of_results}

\begin{center}
\begin{tabular}{ |p{4cm}||p{3cm}|p{3cm}| }
 \hline
 
 \rowcolor{lightgray}
 \multicolumn{3}{|c|}{\textbf{Eigenvalue Estimation}} \\
 \hline
    Model/Error guarantee & Entry samples & Authors \\
    \hline
    \makecell{Bounded entries/\\Additive $\eps n$} & $O(\frac{\poly\log n}{\eps^6})$ & \cite{bhattacharjee2024sublinear} \\
    \, & $O(\frac{\poly\log n}{\eps^5})^*$ & \cite{bhattacharjee2024sublinear} \\
    \, & $\tilde{O}(\frac{1}{\eps^4})$ & Ours \\
    
    \hline
    \makecell{Row-norm sampling/\\Additive $\eps \norm{A}{F}$} & $O(\frac{\poly\log n}{\eps^{16}})$ & \cite{bhattacharjee2024sublinear} \\
    \, & $\tilde{O}(\frac{\poly\log n}{\eps^4})$ & Ours \\

 \hline
\end{tabular}
\caption{Comparison of our eigenvalue sampling results to prior work. The * indicates that the algorithm does not query a principal submatrix.  All algorithms here are non-adaptive.}
\end{center}
\end{table}

%% file: uniform_sampling.tex
\begin{algorithm}
\caption{Uniform Sampling}
\label{alg:uniform_sampling}
\begin{enumerate}
\item Input: Symmetric matrix $A\in \R^{n\times n}$ with $\norm{A}{\infty}\leq 1$, expected sample size $s$
\item If $s\geq n$, let $S = I_n$
\item Otherwise, let $S\in \R^{k\times n}$ be a (rescaled) sampling matrix which samples each row of $A$ independently with probability $p= s/n$ and rescales each sampled row by $\sqrt{\frac{1}{p}}.$
\item Return the $k$ eigenvalues of $S^T A S$, along with $n-k$ additional $0$'s.
\end{enumerate}
\end{algorithm}

%% file: row_norm_sampling.tex
\begin{algorithm}
\caption{Row-norm Sampling}
\label{alg:row_norm_sampling}
\begin{enumerate}
\item Input: Symmetric matrix $A\in \R^{n\times n}$ along with its row norms, expected sample size $s$

\item Let $S\in \R^{k\times n}$ be a (rescaled) sampling matrix which for all $i\in [n]$ samples  $\text{Binomial}(s, \frac{\norm{A_i}{2}^2}{\norm{A}{F}^2})$ copies of row $i$ and rescales those rows by $\frac{1}{\sqrt{p_i}}$ where $p_i := \frac{s\norm{A_i}{2}^2}{\norm{A}{F}^2}.$

\item Implicitly form the matrix $A' \in \R^{n\times n}$ defined by
\[(A')_{ij} = 
\begin{cases} 
      0 & i=j\\
      0 & i\neq j\,\,\text{and}\,\,\norm{A_i}{2}^2\norm{A_j}{2}^2 \leq \frac{\eps^2 \norm{A}{F}^2\abs{A_{ij}}^2}{c\log^4 n}\\
      A_{ij} & \text{otherwise}
\end{cases}
\]

\item Return the $k$ eigenvalues of $S^T A' S$, along with $n-k$ additional $0$'s.
\end{enumerate}
\end{algorithm}

%% file: top_eigvect_algorithm.tex
\begin{algorithm}
\caption{Top Eigenvector Estimation}
\label{alg:top_eigenvector}
\begin{enumerate}
\item \textbf{Input:} Symmetric PSD matrix $A\in \R^{n\times n}$ with $\norm{A}{\infty}\leq 1$, column sampling probability $p$
\item \textbf{Output:} An approximate top eigenvector $u\in \R^n$

\item Sample $S \in \R^{n\times m}$, a column sampling matrix that independently samples each column of $A$ with probability $p$

\item Given $S$, compute $AS$ by sampling the requisite columns of $A$

\item Use $AS$ and $S$ to compute the matrices $S^T A S$ and $S^T A^2 S.$

\item Find $x\in \R^n$ that maximizes $\frac{x^T S^T A^2 S x}{x^T S^T A S x}$

\item Return $\frac{ASx}{\norm{ASx}{2}}$

\end{enumerate}
\end{algorithm}

%% file: technical_sections.tex
\section{Preliminaries}
\subsection{Notation}
Throughout, we use $c$ to denote an absolute constant, which may change between uses. The notation $\tilde{O}(f)$ means $O(f\log^c f)$ for some absolute constant $c.$

Unless otherwise stated, the matrix $A$ always refers to a symmetric (not necessarily PSD) matrix.  We use the notation $A_{\langle k \rangle}$ to denote optimal rank $k$ approximation to $A.$ That is, $A_{\langle k \rangle}$ minimizes $\norm{A - A_{\langle k \rangle}}{F}$ over rank $k$ matrices, where $\norm{A}{F}$ is the Frobenius norm.  In other words, $A_{\langle k \rangle}$ is $A$ with all but its $k$ largest magnitude eigenvalues zeroed out.  The notation $A_{\langle -k \rangle}$ means $A - A_{\langle k \rangle}.$  For a matrix $A,$ the notation $A_i$ denotes the $i$th row of $A.$  A norm without subscripts is always the $\ell_2$ norm for vectors and the operator norm for matrices.
We sometimes denote the eigenvalues of a matrix $A\in\R^{n\times n}$ as $\lambda_1(A) \geq \lambda_2(A) \geq \lambda_n(A).$  The notation $\lambda_{\max}(A)$ is equivalent to $\lambda_1(A).$ We use $\norm{A}{\infty}$ to denote the largest magnitude of an entry in $A.$  When we say that $A$ has ``bounded entries" we will always mean that $\norm{A}{\infty}\leq 1.$ 
 \subsection{Basic Definitions}

We recall the definition of a subspace embedding embedding in the form that we will use it here.  This definition is standard in sketching literature (see \cite{woodruff2014sketching} for example).
\begin{Definition}
We say that $S\in\R^{k\times n}$ is an $1\pm \eps$ distortion subspace embedding for $X\in \R^{n\times m} $ if for all $v\in\R^m$,
\[
(1-\eps)\norm{Xv}{}^2 \leq \norm{SXv}{}^2 \leq (1+\eps) \norm{Xv}{}^2.
\]
\end{Definition}

We will also use the notion of leverage scores early on in order to show that the subspace embedding guarantee that we require is satisfied.  We recall the basic definition here.
\begin{Definition}
Let $X \in \R^{n\times d}$ be a matrix.  The leverage score for row $i$ of $X$ is defined by
\[
\tau_i = e_i^T X (X^T X)^{\dagger}X^T e_i,
\]
where $(X^T X)^{\dagger}$ denotes the Moore-Penrose pseudo-inverse of $X^T X$ and $e_i \in \R^n$ is the $i$th standard basis vector.
\end{Definition}

\subsection{Incoherence Bound}
We use an incoherence bound from \cite{bhattacharjee2024sublinear} in order to control the leverage scores associated to the top eigenvectors. Roughly these bounds say that eigenvectors corresponding to large eigenvalues must be spread out, and moreover that the supports of such eigenvectors cannot overlap much with one another.  For completeness, we supply a short proof here.

\begin{lemma}
\label{lem:incoherence_bound}
Let $A$ be a matrix and let $A_{o}$ denote the projection of $A$ onto the eigenvectors with eigenvalue at least $\alpha.$  Write $A_o = V_o \Lambda_o V_o^T$ where $V_o$ has the eigenvectors of $A$ with eigenvalue at least $\alpha$ as its columns, and $\Lambda_o$ is diagonal matrix containing the associated eigenvalues. Then 
\[
\norm{(V_{o})_i}{}^2 \leq \frac{\norm{A_i}{}^2}{\alpha^2}.
\]
\end{lemma}
\begin{proof}
We have $AV_o = V_o \Lambda_o.$  Therefore
\[
\norm{(AV_o)_i}{}^2 = \norm{(V_o \Lambda_o)_i}{}^2 \geq \alpha^2 \norm{(V_{o})_i}{}^2,
\]
since the diagonal entries of $\Lambda_o$ are all at least $\alpha$ in magnitude. On the other hand,
\[
\norm{(AV_o)_i}{}^2 = \norm{A_i V_o}{}^2 \leq \norm{A_i}{}^2,
\]
since $V_o$ has orthonormal columns.  The lemma follows from combining these two bounds.
\end{proof}

\subsection{Leverage Score Sampling}



We use the following leverage score sampling bound, which is a consequence of a matrix Chernoff bound \cite{kyng2018tutorial, tropp2012user}.

\begin{theorem}
\label{thm:lev_score_sampling}
Let $V\in \R^{n\times d}$, and form a sampling matrix $S$ by including each index $i$ with probability $p_i$ and rescaling by $\frac{1}{\sqrt{p_i}}.$

Let $\tau_i$ denote the leverage score for row $i$ of $V.$ If $p_i \geq \min(\frac{2}{\eps^2} \tau_i\log\frac{d}{\delta}, 1)$, then with probability at least $1 - \delta$ we have
$(1-\eps)V^T V \preceq (SV)^T (SV) \preceq (1+\eps) V^T V.$
\end{theorem}

\subsection{Approximate matrix product}

The following observation was Lemma 1 of \cite{cohen2015optimal}. We supply a short proof here for completeness.
\begin{lemma}
\label{lem:subspace_emb_to_apx_matrix_prod}
Let $A$ and $B$ be real matrices and let $[A| B]$ be the concatenation of $A$ and $B$.  Let $S$ be an $\eps$-distortion subspace embedding for $[A|B].$ Then 
\[
\norm{A^T S^T S B - A^T B}{} \leq \eps \norm{A}{}\norm{B}{}.
\]
\end{lemma}
 
\begin{proof}
Let $x$ and $y$ be unit vectors of appropriate dimension such that
\[
\norm{A^T S^T S B - A^T B}{} 
= x^T (A^T S^T S B - A^T B) y
= \inner{SAx}{SBy} - \inner{Ax}{By}.
\]

Since $S$ is a subspace embedding for $[A|B]$, it preserves the inner product $\inner{Ax}{By}$ up to additive error:
\[
\abs{\inner{SAx}{SBy} - \inner{Ax}{By}}
\leq \eps \norm{Ax}{}\norm{By}{}
\leq \eps \norm{A}{}\norm{B}{}.
\]
The lemma follows.
\end{proof}

Combining with Theorem~\ref{thm:lev_score_sampling} above immediately gives a generalization of Theorem~\ref{thm:lev_score_sampling} for constructing approximating matrix products.  This will account for most of our use cases for leverage score sampling guarantees, so we record this fact here for easy reference.

\subsection{Additional facts}
\begin{proposition}
\label{prop:row_norms_stay_bounded}
Let $A$ be a symmetric matrix.  Let $P$ be a projection onto the span of a subset of $A$'s eigenvectors.  Then $\norm{(PAP^T)_k}{} \leq \norm{A_k}{}$ where $A_k$ is the $k$th row of $A.$
\end{proposition}
\begin{proof}
Using the spectral decomposition of $A$, we can write $A = \sum_i u_i u_i^T$ and $PAP^T = \sum_j v_j v_j^T$ where the $u_i$'s are mutually orthogonal, and where the $v_j$'s are a subset of the $u_i$'s.  Then 
\begin{align*}
\norm{A_k}{}^2 
&= \norm{Ae_k}{}^2 \\
&= \norm{\sum_i \inner{u_i}{e_k} u_i}{}^2\\
&= \sum_i \norm{u_i}{}^2 \abs{\inner{u_i}{e_k}}^2\\ 
&\geq \sum_j \norm{v_j}{}^2 \abs{\inner{v_j}{e_k}}^2\\
&= \norm{(PAP^T)_k}{}^2.
\end{align*}
\end{proof}

\section{Our subspace embedding condition}

In the next section we obtain both upper and lower bounds on the outlying eigenvalues.  Importantly, we show how to obtain these bounds using only subspace embedding guarantees.  The leverage score estimates will show that we can use a sampling procedure to construct the desired subspace embeddings.

We will let $S$ be our sampling matrix throughout.  The key subspace embedding property that we require is as follows.

\begin{assumption}
\label{assump:subspace_embedding}
For all $\lambda \geq \outlyingthresh$, $S$ is a $\min(\outlyingthresh/\lambda, 1/10)$ distortion subspace embedding on the span of eigenvectors of $A$ with associated eigenvalue at least $\lambda$ in magnitude
\end{assumption}

 To streamline our analysis, we use this assumption in full generality. One should think of $\outlyingthresh$ as representing the threshold defining the outlying eigenvalues. We will later specialize to $\outlyingthresh = \eps n$ and $\outlyingthresh = \eps \norm{A}{F}$.  Our first step is to show that Assumption \ref{assump:subspace_embedding} is realized by our sampling schemes.  Then we treat this property as a black-box for the remainder of the argument, and show that it implies an additive $O(\outlyingthresh)$ spectral approximation guarantee.

 \subsection{Realizing Assumption \ref{assump:subspace_embedding}}

 As a consequence of the incoherence bound, and leverage score sampling guarantee, we show that a version of Assumption~\ref{assump:subspace_embedding} is achieved both for uniform sampling and row sampling.  For uniform sampling we will assume that $A$ has bounded entries and set $\outlyingthresh = \eps n.$  For row-norm sampling, we will make no assumptions on $A$ and set $\outlyingthresh = \eps\norm{A}{F}$.  In both cases, we will need our sample to have size roughly $\frac{1}{\eps^2}.$
 
\begin{lemma}
    \label{lem:subspace_embedding_from_uniform_sampling}
     Let $A\in\R^{n\times n}$ have all entries bounded in magnitude by $1$.  Let $S$ be a sampling matrix which samples each row of $A$ with probability $s/n$.  
     Then for $s\geq \frac{c}{\eps^2}\left(\log\log\frac{1}{\eps} + \log\frac{1}{\eps^2 \delta}\right)$, $S$ satisfies Assumption \ref{assump:subspace_embedding} with $\outlyingthresh = \eps n.$
\end{lemma}

\begin{proof}
Consider a fixed value of $\lambda$, and let $V_{\geq\lambda} \in \R^{n\times d}$ be an orthonormal matrix whose columns are the eigenvectors of $A$ with associated eigenvalue at least $\lambda.$

By Lemma~\ref{lem:incoherence_bound}, 
\[
\norm{V_{\geq\lambda,i}}{}^2 \leq \frac{1}{\lambda^2}\norm{A_i}{}^2 \leq \frac{n}{\lambda^2},
\]
since $A$ has bounded entries.  Since $V_{\geq\lambda,i}$ has orthonormal columns, this says that all leverage scores of $V_{\geq\lambda}$ are at least $n/\lambda^2.$

Let $p = s/n,$ and set $\eps = \frac{1}{2}(\outlyingthresh/\lambda)$ in Theorem~\ref{thm:lev_score_sampling}.  This shows that if 
\[p\geq 8 \frac{\lambda^2}{\outlyingthresh^2} \left(\frac{n}{\lambda^2}\right)\log\frac{d}{\delta}
= 8\frac{n}{\outlyingthresh^2} \log\frac{d}{\delta}
= \frac{8}{\eps^2 n} \log \frac{d}{\delta},
\]
then $S$ is a a $\frac{1}{2}(\outlyingthresh/\lambda)$ distortion subspace embedding for $V_{\geq\lambda}$, with probability at least $1-\delta.$  Note that $\norm{A}{F}^2 \leq n^2$, so $A$ has at most $1/\eps^2$ eigenvalues that are at least $\eps n.$  This implies that $d\leq \frac{1}{\eps^2}$ so it suffices to have
$
p \geq \frac{8}{\eps^2 n} \log\frac{1}{\eps^2 \delta},
$
and hence to choose $s \geq \frac{8}{\eps^2}\log\frac{1}{\eps^2\delta}.$

In order for Assumption~\ref{assump:subspace_embedding} to hold, it suffices to have $S$ be a $1/10$ distortion embedding for $V$, a $\frac{1}{2}$ distortion embedding for $V_{\geq \eps n}$, a $\frac{1}{4}$ distortion embedding for $V_{\geq 2\eps n}$, and in general a $\frac{1}{2^{r+1}}$ distortion embedding for $V_{\geq 2^r \eps n}$ for $r = 0, \ldots, \lfloor \log_2\frac{1}{\eps}\rfloor.$  By the same calculations as above, we can achieve the $1/10$ distortion embedding for $V$ using a sampling probability of $p \geq \frac{c}{\eps^2 n}\log\frac{1}{\delta}.$
Then replacing $\delta$ by $\delta/\lfloor \log_2\frac{1}{\eps}\rfloor$ and taking a union bound over all $r$ yields the lemma.

\end{proof}

The following lemma gives the analogous result for squared row-norm sampling.  For technical reasons, we give a version that will allow some entries of $A$ to be zeroed out prior to sampling, although this does not change much.
\begin{lemma}
    \label{lem:subspace_embedding_from_row_norm_sampling}
     Let $A\in\R^{n\times n}$ be arbitrary.  Let $S$ be a sampling matrix which samples each row index $i$ with probability $p_i = \min\left(s\frac{\norm{A_i}{}^2}{\norm{A}{F}^2}, 1\right)$.  
     Let $A'$ be the matrix $A$, but possibly with some of its entries replaced with $0.$
     Then for $s\geq \frac{c}{\eps^2}\left(\log\log\frac{1}{\eps} + \log\frac{1}{\eps^2 \delta}\right)$, $S$ satisfies Assumption \ref{assump:subspace_embedding} for $A'$ with $\outlyingthresh = \eps \norm{A}{F}.$
\end{lemma}

\begin{proof}
Consider a fixed value of $\lambda$, and let $V'_{\geq\lambda} \in \R^{n\times d}$ be an orthonormal matrix whose columns are the eigenvectors of $A'$ with associated eigenvalue at least $\lambda.$

By Lemma~\ref{lem:incoherence_bound}, 
\[
\norm{V'_{\geq\lambda,i}}{}^2 \leq \frac{1}{\lambda^2}\norm{A'_i}{}^2
\leq \frac{1}{\lambda^2}\norm{A_i}{}^2.
\]

Set $\eps = \frac{1}{2}(\outlyingthresh/\lambda) = \frac{1}{2}\frac{\eps\norm{A}{F}}{\lambda}$ in Theorem~\ref{thm:lev_score_sampling}.  This shows that if 
\[p_i 
\geq \min\left(8\frac{\lambda^2}{\outlyingthresh^2}\left(\frac{\norm{A_i}{}^2}{\lambda^2}\right)\log\frac{d}{\delta}, 1\right)
= \min\left(\frac{8\norm{A_i}{}^2}{\eps^2 \norm{A}{F}^2}\log\frac{d}{\delta}, 1\right),
\]
then $S$ is a a $\frac{1}{2}(\outlyingthresh/\lambda)$ distortion subspace embedding for $V'_{\geq\lambda}$, with probability at least $1-\delta.$  Note that $A'$ has at most $1/\eps^2$ eigenvalues that are at least $\eps \norm{A}{F}.$  This implies that $d\leq \frac{1}{\eps^2}$ so it suffices to have
\[
p_i \geq \min\left(\frac{8\norm{A_i}{}^2}{\eps^2 \norm{A}{F}^2} \log\frac{1}{\eps^2 \delta}, 1\right),
\]
and hence to choose $s \geq \frac{8}{\eps^2}\log\frac{1}{\eps^2\delta}.$

In order for Assumption~\ref{assump:subspace_embedding} to hold, it suffices to have $S$ be a $1/10$ distortion embedding for $V'$, a $\frac{1}{2}$ distortion embedding for $V_{\geq \eps n}$, a $\frac{1}{4}$ distortion embedding for $V'_{\geq 2\eps n}$, and in general a $\frac{1}{2^{r+1}}$ distortion embedding for $V'_{\geq 2^r \eps n}$ for $r = 0, \ldots, \lfloor \log_2\frac{1}{\eps}\rfloor.$  By the same calculations as above, we can achieve the $1/10$ distortion embedding for $V'$ using a sampling probability of $p \geq \frac{c}{\eps^2 n}\log\frac{1}{\delta}.$
Then replacing $\delta$ by $\delta/\lfloor \log_2\frac{1}{\eps}\rfloor$ and taking a union bound over all $r$ yields the lemma.
\end{proof}

\section{Outlying Eigenvalue Bounds}

Using Assumption~\ref{assump:subspace_embedding} we show how to obtain both upper and lower bound on the outlying eigenvalues of $A.$ 

\subsection{Lower Bounds}

\begin{proposition}
\label{prop:principle_angle_fact}
Let $W_1$ and $W_2$ be subspaces of $\R^n$, and let $P$ be the orthogonal projection onto $W_2.$  Let $x$ be a nonzero vector in $W_1.$ Then
\[
\frac{\norm{Px}{}}{\norm{x}{}} \leq \max_{w_1\in W_1, w_2 \in W_2} \frac{\inner{w_1}{w_2}}{\norm{w_1}{}\norm{w_2}{}}. 
\]
\end{proposition}
\begin{proof}
For $w_1\in W_1$ and $w_2 \in W_2$ we have
\[
\max_{w_1, w_2}\frac{\inner{w_1}{w_2}}{\norm{w_1}{}\norm{w_2}{}}
= \max_{w_1}\frac{\inner{w_1}{Pw_1}}{\norm{w_1}{}\norm{Pw_1}{}}
\geq \frac{\inner{x}{Px}}{\norm{x}{}\norm{Px}{}} = \frac{\norm{Px}{}}{\norm{x}{}}.
\]
\end{proof}

\begin{proposition}
\label{prop:rayleigh_quotient_perturbation}
Let $A$ be symmetric and $x$ and $h$ be vectors.  Suppose that $\norm{x}{} \geq 1-\alpha$ and $\norm{h}{} \leq \alpha$ with $\alpha \leq 1/4.$  Then 
\[
\bigg|\frac{(x+h)^T A (x+h)}{\norm{x+h}{}^2} - \frac{x^T A x}{\norm{x}{}^2} \bigg| 
\leq 8 \norm{A}{} \alpha.
\]
\end{proposition}
\begin{proof}
Let $F(x) = \frac{x^T A x}{x^T x}$ be the Rayleigh quotient.  Recall that its gradient is given by
\[
\nabla F(v) = \frac{2}{\norm{v}{}^2}\left(A - \frac{v^T A v}{\norm{v}{}^2}I\right) v.
\]
Then
\[
\norm{\nabla F(v)}{} 
\leq \frac{2}{\norm{v}{}^2}\norm{A - \frac{v^T A v}{\norm{v}{}^2} I}{}\norm{v}{}
\leq \frac{2}{\norm{v}{}}\left(\norm{A}{} + \norm{\frac{v^T A v}{\norm{v}{}^2}I}{}\right)
\leq 4\frac{\norm{A}{}}{\norm{v}{}}.
\]
This together with the hypotheses implies that $F$ is $(8\norm{A}{})$-Lipschitz on a $\norm{h}{}$-neighborhood of $x$, and the claim follows.
\end{proof}

\begin{lemma}
\label{lem:SPiv_lower_bound}
Let $\Pi$ have orthonormal columns and suppose that $S$ is an $\alpha$-distortion subspace embedding for $\Pi.$ Then for all nonzero $v,$
\[
\frac{(S\Pi v)^T S\Pi\Pi^T S^T (S\Pi v)}{\norm{S\Pi v}{}^2}
\geq 1 - \alpha.
\]
\end{lemma}
\begin{proof}
The numerator above is equal to $\norm{\Pi^T S^T S \Pi v}{}^2$.  By Cauchy-Schwarz,
\[
\norm{\Pi^T S^T S \Pi v}{}^2
\geq \frac{\inner{v}{\Pi^T S^T S \Pi v}^2}{\norm{v}{}^2}
= \frac{\norm{S\Pi v}{}^4}{\norm{v}{}^2}.
\]
So the fraction in the lemma is at least
$
\frac{\norm{S\Pi v}{}^2}{\norm{v}{}^2}.
$
Since $S$ is an $\alpha$ distortion subspace embedding for $\Pi,$ we have
\[
\norm{S\Pi v}{}^2 \geq (1 - \alpha)\norm{\Pi v}{}^2
= (1-\alpha)\norm{v}{}^2,
\]
and the lemma follows.

\end{proof}

\begin{lemma}
\label{lem:projecting_doesnt_hurt}
Let $U$ and $\Pi$ have mutually orthonormal columns and let $M$ be symmetric.  Suppose that $S$ is an $\alpha$-distortion embedding for $[U|\Pi],$ where $\alpha\leq 1/7.$ Let $P$ be the orthogonal projection onto the column space of $SU$ and let $P^{\perp} = I - P.$  Then for all vectors $x$, $P^{\perp} S\Pi x$ is nonzero and
\[
\bigg|
\frac{(P^{\perp}S\Pi x)^T M (P^{\perp}S\Pi x)}{\norm{P^{\perp}S\Pi x}{}^2}
- 
\frac{(S\Pi x)^T M (S\Pi x)}{\norm{S\Pi x}{}^2}
\bigg|
\leq 10\alpha\norm{M}{}.
\]
\end{lemma}
\begin{proof}
We first show that $P^{\perp}S\Pi x$ is close to $S\Pi x.$  By Proposition~\ref{prop:principle_angle_fact} we have
\[
\frac{\norm{P S\Pi x}{}}{\norm{S\Pi x}{}} \leq \max_{y,z} \frac{\inner{S\Pi y}{SUz}}{\norm{S\Pi y}{}\norm{SUz}{}}
= \max_{y,z}\frac{y^T \Pi^T S^T S U z}{\norm{S\Pi y}{}\norm{S U z}{}}
\leq \max_{y,z} \frac{\norm{y}{}\norm{z}{}\norm{\Pi^T S^T S U}{}}{\norm{S\Pi y}{}\norm{S U z}{}}.
\]
By our hypothesis that $S$ is a subspace embedding, along with Lemma~\ref{lem:subspace_emb_to_apx_matrix_prod},we bound the operator norm in the numerator by
\[
\norm{\Pi^T S^T S U}{}
= \norm{\Pi^T S^T S U - \Pi^T U}{} \leq \alpha \norm{\Pi}{}\norm{U}{} = \alpha.
\]
For the denominator, 
the subspace embedding property implies that 
\[
(1-\alpha)^{1/2}\norm{y}{}\leq \norm{S\Pi y}{} \leq (1 + \alpha)^{1/2}\norm{y}{}
\]
and similarly for $\norm{SUz}{}$ and $\norm{S\Pi x}{}.$ Plugging into the bound above gives that for all unit vectors $x,$
\[
\norm{P S\Pi x}{}
\leq \frac{\alpha}{(1 - \alpha)^{1/2}(1-\alpha)^{1/2}} \norm{S\Pi x}{}
\leq \frac{\alpha}{1-\alpha} (1+\alpha)^{1/2}
< 1.25\alpha,
\]
for $\alpha \leq 1/7.$
By the triangle inequality,
\[
0 = \norm{PS\Pi x + P^{\perp}S\Pi x - S\Pi x}{}
\geq \norm{P^{\perp}S\Pi x - S\Pi x}{} - \norm{PS\Pi x}{},
\]
and so $\norm{P^{\perp}S\Pi x - S\Pi x}{} \leq 1.25\alpha$ for all unit vectors $x.$

Note that $\norm{S\Pi x}{}\geq (1-\alpha)^{1/2}\norm{x}{} \geq (1-\alpha){}.$
To apply Proposition~\ref{prop:rayleigh_quotient_perturbation} we need $1.25\alpha\leq 1/4$, which is true by hypothesis.

So by Proposition~\ref{prop:rayleigh_quotient_perturbation} we have 
\[
\bigg|
\frac{(P^{\perp}S\Pi x)^T M (P^{\perp}S\Pi x)}{\norm{P^{\perp}S\Pi x}{}^2}
-
\frac{(S\Pi x)^T M (S\Pi x)}{\norm{S\Pi x}{}^2}
\bigg|
\leq (8\cdot 1.25)\alpha \norm{M}{}
= 10\alpha\norm{M}{},
\]
as desired.
\end{proof}

\begin{lemma}
\label{lem:eigval_lower_bound}
Let $A \in \R^{n\times n}$ be symmetric, not necessarily PSD.  Assume that all nonzero eigenvalues of $A$ are at least $\outlyingthresh$ in magnitude. 
Suppose that $S$ satisfies Assumption~\ref{assump:subspace_embedding}.
Then for all $k$ with $\lambda_k(A) > 0$ we have
$\lambda_k(SAS^T) \geq \lambda_k(A) - 51\outlyingthresh$.

\end{lemma}

\begin{proof}

Suppose that the $k$th largest eigenvalue of $A$ is $\lambda>0.$  We will show that $S A S^T$ is at least $\lambda - O(\outlyingthresh)$ on a $k$-dimensional subspace $W$, which will imply that $\lambda_k(SAS^T) \geq \lambda - O(\outlyingthresh).$

Towards this end, decompose $A = A_+ - A_-$ into its positive and negative parts.  Since $A_+$ has $k$ eigenvalues that are at least $\lambda$, we can write $A_+ \succeq \lambda \Pi\Pi^T$ where the columns of $\Pi$ are the $k$ largest eigenvectors of $A_+$ (chosen to be orthonormal when there are duplicate eigenvalues).  Let $U_{\geq \lambda}$ have as its columns the eigenvectors of $A_-$ with associated eigenvalue at least $\lambda,$ and similarly for $U_{<\lambda}$. Let $P$ be the orthogonal projection onto the column span of $SU_{\geq\lambda}$ and set $P^{\perp} = I - P.$ As suggested by the previous lemma, we choose our subspace $W$ to be the column span of $P^{\perp} S\Pi.$

Write $f(\lambda) := \min(1/7, \outlyingthresh/\lambda)$ and recall that $S$ is an $f(\lambda)$-distortion subspace on the span of the eigenvectors of $A$ with associated eigenvalue at least $\lambda$ in magnitude, in other words on $[\Pi|U_{\geq\lambda}]$. Lemma~\ref{lem:projecting_doesnt_hurt} applies and gives that for all $x,$
\begin{align*}
\frac{(P^{\perp}S\Pi x)^T (S\Pi\Pi^TS^T) (P^{\perp} S \Pi x)}{\norm{P^{\perp} S\Pi x}{}^2}
&\geq \frac{(S\Pi x)^T (S\Pi\Pi^TS^T) (S \Pi x)}{\norm{S\Pi x}{}^2} - 10f(\lambda)\norm{S\Pi\Pi^TS^T}{}\\
&\geq (1 - f(\lambda)) - 10f(\lambda)\norm{S\Pi\Pi^T S^T}{}\\
&\geq 1 - 12f(\lambda),
\end{align*}
where we used Lemma~\ref{lem:SPiv_lower_bound} as well as as the observation that 
\[
\norm{S\Pi\Pi^T S^T}{} = \norm{\Pi^T S^T S \Pi}{} \leq (1 + f(\lambda))^{1/2} \leq (1+1/7)^{1/2} \leq 1.07,
\]
by the subspace embedding property.  Also note that the dimension of $W$ is $k$ since Lemma~\ref{lem:projecting_doesnt_hurt} also states that $P^{\perp} S\Pi$ has trivial kernel. Now for $w\in W$, with $\norm{w}{} = 1$ we have 
\[
w^T S A_+ S^T w
\geq \lambda w^T S\Pi \Pi^T S^T w
= \lambda\norm{\Pi^T S^T w}{}^2
\geq \lambda(1 - f(\lambda)).
\]

Next we study $w^T A_- w$ where $w\in W$ is a unit vector.  We first split the spectral decomposition of $A_-$ into two pieces as follows:
\[
A_- = U_{\geq\lambda}\Lambda_{\geq\lambda}U_{\geq\lambda}^T + U_{<\lambda}\Lambda_{<\lambda}U_{<\lambda}^T.
\]

Then we have
\begin{align}
\label{eq:A_minus_level_set_decomp}
w^T (SA_- S^T)w 
&=  w^T SU_{\geq\lambda}\Lambda_{\geq\lambda}U_{\geq\lambda}^T S^T w + w^T S U_{<\lambda}\Lambda_{<\lambda}U_{<\lambda}^T S^T w\\
&= w^T S U_{<\lambda}\Lambda_{<\lambda}U_{<\lambda}^T S^T w,
\end{align}
where the second inequality is because $w$ is orthogonal to the image of $SU_{\geq \lambda}$, which comes from the definition of $W.$

We now partition the eigenvalues into approximate level sets.  Let $V_k$ have as its columns the columns of $U_{<\lambda}$ with associated eigenvalue in $[2^{k-1}\outlyingthresh , 2^k \outlyingthresh).$  Then we have the spectral bound
\[
U_{\leq \lambda} \Lambda_{\leq\lambda}U_{\leq \lambda}^T \preceq (2\outlyingthresh) V_1 V_1^T + (4\outlyingthresh)V_2V_2^T + \ldots + (2^r \outlyingthresh) V_r V_r^T,
\]
where $r$ is $\lceil\log_2(\lambda/\outlyingthresh)\rceil.$ 

We would like to bound $w^T S V_kV_k^T S^T w.$  Write $w = P^{\perp} S\Pi x$ for some $x.$  Then similar to our application of Lemma~\ref{lem:projecting_doesnt_hurt} above,
\[
\frac{(P^{\perp}S\Pi x)^T (SV_k V_k^T S^T)(P^{\perp}S\Pi x)}{\norm{P^{\perp}S\Pi x}{}^2}
\leq 
\frac{(S\Pi x)^T (SV_k V_k^T S^T)(S\Pi x)}{\norm{S\Pi x}{}^2} + 10f(\lambda)\norm{SV_kV_k^T S^T}{}.
\]
Since $S$ a $1/7$ distortion embedding for $A$ and in particular $V_k,$ we have 
\[
\norm{SV_k V_k^T S^T}{}
= \norm{V_k^T S^T S V_k}{}
= \max_{\norm{x}{}=1} \norm{SV_k x}{}^2
\leq \frac{8}{7}.
\] For the first term above,
\[
\frac{(S\Pi x)^T (SV_k V_k^T S^T)(S\Pi x)}{\norm{S\Pi x}{}^2}
=
\frac{\norm{V_k^T S^T S \Pi x}{}^2}{\norm{S\Pi x}{}^2}
\leq
\frac{\norm{V_k^T S^T S \Pi}{}^2\norm{x}{}^2}{\norm{S\Pi x}{}^2}.
\]
Since $S$ is an $f(2^{k-1}\outlyingthresh)$-distortion subspace embedding for $[V_k | \Pi],$ we have 
\[
\norm{V_k^T S^T S\Pi}{} \leq f(2^{k-1}\outlyingthresh).
\]
Also $\norm{S\Pi x}{}^2 \geq \frac{6}{7} \norm{x}{}^2$ again since $S$ is a subspace embedding for $\Pi$, and so 
\[
\frac{(P^{\perp} S\Pi x)^T (SV_k V_k^T S^T)(P^{\perp}S\Pi x)}{\norm{P^{\perp}S\Pi x}{}^2} \leq \frac{7}{6}f(2^{k-1}\outlyingthresh)^2 + 12 f(\lambda).
\]
In particular, for $w\in W$ a unit vector,
\[
w^T (SV_k V_k^T S^T) w \leq \frac{7}{6}f(2^k\outlyingthresh)^2 + 12f(\lambda).
\]
Plugging into \eqref{eq:A_minus_level_set_decomp} above gives
\begin{align*}
w^T (S A_- S^T) w
&= w^T S U_{\leq \lambda} \Lambda_{\leq\lambda}U_{\leq \lambda}^T S^T w\\
&\leq w^T S \left((2\outlyingthresh) V_1 V_1^T + (4\outlyingthresh)V_2V_2^T + \ldots + (2^r \outlyingthresh) V_r V_r^T\right) S^T w\\
&= \sum_{k=1}^r (2^k\outlyingthresh) w^T SV_k V_k^T S^T w\\
&\leq \sum_{k=1}^r(2^k\outlyingthresh)\left(\frac{7}{6}f(2^k\outlyingthresh)^2 + 12f(\lambda)\right).
\\
&\leq \sum_{k=1}^r (2^k \outlyingthresh)\left(\frac{7}{6}(\frac{\outlyingthresh}{2^k\outlyingthresh})^2 + 12\frac{\outlyingthresh}{\lambda}\right)\\
&= \outlyingthresh \sum_{k=1}^r \left(\frac{7}{6}2^{-k} + 2^k\frac{12\outlyingthresh}{\lambda}\right)\\
&\leq \frac{7}{6}L + 12L\sum_{k=1}^r \frac{2^k L}{\lambda}\\
&\leq \frac{7}{6}L + 48L, 
\end{align*}
since $2^r \outlyingthresh \leq 2\lambda.$

Putting the pieces together,
\begin{align*}
w^T S A S^T w
&= w^T S A_+ S^T w - w^T S A_- S^T w\\
&\geq \lambda(1-f(\lambda)) - 50\outlyingthresh\\
&\geq \lambda\left(1 - \frac{\outlyingthresh}{\lambda}\right) - 50\outlyingthresh\\
& = \lambda - 51\outlyingthresh,
\end{align*}
for all unit vectors $w\in W.$  Since the dimension of $W$ is $k$, this implies that $\lambda_k(SAS^T) \geq \lambda - 51\outlyingthresh$ as desired. 
\end{proof}

\subsection{Upper Bounds}

\begin{lemma}
Suppose that the nonzero eigenvalues of $A$ are all at least $\outlyingthresh$ and let $\lambda_1 = \lambda_{\max}(A).$  Suppose that $S$ satisfies Assumption~\ref{assump:subspace_embedding}.  Then $\lambda_{\max}(SAS^T) \leq \lambda_1 + c\outlyingthresh\log\frac{\lambda_1}{L}.$
\end{lemma}

\begin{proof}
Let $A = U\Lambda U^T$ be the spectral decomposition of $A.$

We partition the columns of $U$ based on the eigenvalues.  For $i=1,\ldots, r$ with $r = \lceil \log_2\frac{\lambda_1}{\outlyingthresh}\rceil$, let $U^{(i)}$ consist of the columns of $U$ corresponding to eigenvalues in $(\lambda_1 2^{-i}, \lambda_1 2^{-i + 1}].$  Define $\Lambda^{(i)}$ to be the diagonal matrix consisting of the associated eigenvalues.



Let $j\leq i$ be an arbitrary index. Note that if $S$ is a subspace embedding for $U^{(i)}$ then $S$ is also a subspace embedding for $U^{(i)}\Lambda^{(i)}$ with the same parameters.

By Assumption~\ref{assump:subspace_embedding}, $S$ is a $\outlyingthresh/(\lambda_1 2^{-i}) = 2^i \outlyingthresh/\lambda_1$ distortion subspace for $[U^{(i)} | U^{(j)}].$  Therefore we have 

\begin{align*}
\norm{\left(S U^{(i)}(\Lambda^{(i)})^{1/2}\right)^T \left(S U^{(j)}(\Lambda^{(j)})^{1/2}\right) - \left(U^{(i)}(\Lambda^{(i)})^{1/2}\right)^T \left(U^{(j)}(\Lambda^{(j)})^{1/2}\right)}{}
&\leq \frac{2^i \outlyingthresh}{\lambda_1} \norm{U^{(i)}(\Lambda^{(i)})^{1/2}}{} \norm{U^{(j)}(\Lambda^{(j)})^{1/2}}{} \\
&\leq \frac{2^i \outlyingthresh}{\lambda_1}  (\lambda_1 2^{-i +1})^{1/2}(\lambda_1 2^{-j+1})^{1/2}\\
&= 2L \sqrt{2}^{i - j}.
\end{align*}

For $i > j$, $U^{(i)}$ and $U^{(j)}$ are orthognal to one another, so 
\[
\norm{\left(S U^{(i)}(\Lambda^{(i)})^{1/2}\right)^T \left(S U^{(j)}(\Lambda^{(j)})^{1/2}\right)}{} 
\leq 
2L \sqrt{2}^{i-j}.
\]
Therefore for all $i\neq j$,
\[
\norm{\left(S U^{(i)}(\Lambda^{(i)})^{1/2}\right)^T \left(S U^{(j)}(\Lambda^{(j)})^{1/2}\right)}{} 
\leq 2L \sqrt{2}^{\abs{i-j}}.  
\]

Similarly for the $i=j$ case, we have
\[
\norm{\left(S U^{(i)}(\Lambda^{(i)})^{1/2}\right)^T \left(S U^{(i)}(\Lambda^{(i)})^{1/2}\right) - \left(U^{(i)}(\Lambda^{(i)})^{1/2}\right)^T \left(U^{(i)}(\Lambda^{(i)})^{1/2}\right)}{}
\leq 2L, 
\]
which by the triangle inequality implies that
\[
\norm{\left(S U^{(i)}(\Lambda^{(i)})^{1/2}\right)^T \left(S U^{(i)}(\Lambda^{(i)})^{1/2}\right)}{}
\leq \lambda_1 2^{-i + 1} + 2L.
\]

Now let $x$ be a unit vector, and partition the coordinates of $x$ into pieces $x^{(i)}$ corresponding to the decomposition of $U$ into $U^{(i)}$'s. Then
\[
\norm{S U \Lambda^{1/2} x}{}^2
= \norm{\sum_{i} S U^{(i)}(\Lambda^{(i)})^{1/2}x^{(i)} }{}^2
= \sum_i \norm{S U^{(i)}(\Lambda^{(i)})^{1/2}x^{(i)}}{}^2 + 2\sum_{i < j} \inner{S U^{(i)}(\Lambda^{(i)})^{1/2}x^{(i)} }{S U^{(j)}(\Lambda^{(j)})^{1/2}x^{(j)} } 
\]

By the bounds above, the cross terms are all bounded by $\outlyingthresh \sqrt{2}^{\abs{i-j}} \norm{x^{(i)}}{} \norm{x^{(j)}}{}$ and the diagonal terms are bounded by $(\lambda_1 2^{-i + 1} + \outlyingthresh) \norm{x^{(i)}}{}^2.$

Thus 
\begin{align*}
\norm{S U \Lambda^{1/2} x}{}^2
&\leq \sum_{i=1}^r \left[(\lambda_1 2^{-i+1} + \outlyingthresh)\norm{x^{(i)}}{}^2\right] + 2\sum_{i<j}\outlyingthresh \sqrt{2}^{\abs{i-j}}\norm{x^{(i)}}{}\norm{x^{(j)}}{}\\
&=
\outlyingthresh + \sum_{i=1}^r \lambda_1 2^{-i+1}\norm{x^{(i)}}{}^2 + 2\outlyingthresh\sum_{i < j}\sqrt{2}^{\abs{i-j}} \norm{x^{(i)}}{}\norm{x^{(j)}}{}.
\end{align*}

Let $M$ be an $r\times r$ matrix defined by $M_{1,1} = \lambda_1$, $M_{i,i} = \lambda_1/2$ for $i>1$, and for $i\neq j$ $M_{i,j} = M_{j,i} = \outlyingthresh \sqrt{2}^r := \alpha.$

Comparing the entries of $M$ with the coefficients above gives
\[
\norm{S U \Lambda^{1/2} x}{}^2
\leq \outlyingthresh + \norm{M}{}.
\]
To calculate $\norm{M}{}$, note from the symmetry of the last $r-1$ coordinates that $v^T M v$ is maximized for a $v$ of the form $[a,b/\sqrt{r-1},b/\sqrt{r-1},\ldots b/\sqrt{r-1}],$ where $a^2 + b^2 = 1.$  For this choice of $v$, we have
\[
v^T M v = \lambda_1 a^2 + 2\alpha\sqrt{r-1}ab + (\lambda_1/2 + (r-2)\alpha)b^2. 
\]

So the operator norm of $M$ is the same as the operator norm of 
\[
M' := 
\begin{bmatrix}
\lambda_1& \alpha\sqrt{r-1}\\
\alpha\sqrt{r-1}& \lambda_1/2 + \alpha(r-2)\\
\end{bmatrix}.
\]

Now suppose that $\lambda_1 + T$ is an eigenvalue of $M'.$ From the characteristic equation,
\[
T(T + \lambda_1/2 - \alpha(r-2)) = \alpha^2 (r-1).
\]

We claim that $\lambda_1/2 - \alpha(r-2) \geq \lambda_1/4$ when $\lambda_1\geq C\outlyingthresh$ for an absolute constant $C.$  To see this, note that the desired bound is implied by $\lambda_1/4 \geq \alpha r.$ Note that $\alpha = \outlyingthresh \sqrt{2^r}\leq \sqrt{2} \sqrt{\outlyingthresh\lambda_1},$ so it suffices to check that $\lambda_1/4 \geq \sqrt{2} \sqrt{\outlyingthresh\lambda_1}(\log_2(\lambda_1/\outlyingthresh) + 1),$ or equivalently $\sqrt{\lambda_1/\outlyingthresh} \geq 4\sqrt{2} (\log_2(\lambda_1/\outlyingthresh + 1).$ The latter is true for large enough $C$, since we have the real inequality $\sqrt{x} \geq 4\sqrt{2}(\log_2(x)+1)$ for large enough $x.$ 

Then for $\lambda_1 \geq C\outlyingthresh$ we have
\[
T(T + \lambda_1/2 - \alpha(r-2)) \geq T(T + \lambda_1/4) \geq T\lambda_1/4,
\]
while on the other hand $\alpha^2 (r-1) \leq 2\outlyingthresh \lambda_1 r.$  So we conclude that
\[
T\lambda_1/4 \leq 2\outlyingthresh \lambda_1 r,
\]
which implies that $T \leq 8\outlyingthresh r = 8 L \lceil 
\log_2(\frac{\lambda_1}{\outlyingthresh})\rceil.$ This in turn gives the bound
\[
\norm{S U \Lambda^{1/2} x}{}^2 \leq \lambda_1 + 8 L \lceil 
\log_2(\frac{\lambda_1}{\outlyingthresh})\rceil.
\]
But $x$ was arbitrary so 
\[\norm{S A S^T}{} = \norm{\Lambda^{1/2}U^T S^T S U \Lambda^{1/2}}{}
= \norm{S U \Lambda^{1/2}}{}^2
\leq \lambda_1 + 8 L \lceil 
\log_2(\frac{\lambda_1}{\outlyingthresh})\rceil
\]
whenever $\lambda_1 \geq C\outlyingthresh$. The lemma follows after adjusting constants.



\end{proof}

\begin{lemma}
\label{lem:eigval_upper_bounds}
Let $A\in \R^{n\times n}$ with $\lambda_i(A) > 0.$ 
Suppose that $S$ satisfies Assumption~\ref{assump:subspace_embedding} for $A.$ Then $\lambda_i(S A S^T)\leq \lambda_i(A) + \outlyingthresh\log\frac{\lambda_i(A)}{\outlyingthresh}.$
\end{lemma}

\begin{proof}
First note that we may as well prove this result for PSD $A$, since $\lambda_i(SAS^T) \leq \lambda_i(S A_+ S^T)$, where $A_+$ is the PSD part of $A.$ So we assume from now on that $A$ is PSD. Let $A_{\langle -i \rangle}$ denote $A$ with the top $i$ eigenvalues zeroed out, and let $S_{i}$ denote the orthogonal complement of the span of the top $i-1$ eigenvectors of $A.$ Note that $S_{i}$ has dimension $n-i+1.$ By the min-max theorem,
\begin{align*}
\lambda_i(S A S^T)
&= \lambda_i(A^{1/2}S^T S A^{1/2})\\
&= \min_{\substack{T\subseteq \R^n\\ \dim(T) = n-i+1} }\max_{\substack{x\in T\\ \norm{x}{}=1}} x^T A^{1/2}S^T S A^{1/2} x\\
&\leq \max_{x\in S_{i}, \norm{x}{}=1} x^T A^{1/2}S^T S A^{1/2} x\\
&= \max_{x\in S_{i}, \norm{x}{}=1} x^T A_{-i}^{1/2}S^T S A_{\langle -i \rangle}^{1/2} x\\
&\leq \max_{ \norm{x}{}=1} x^T A_{\langle -i \rangle}^{1/2}S^T S A_{-i}^{1/2} x\\
&= \lambda_{\max}(A_{\langle -i \rangle}^{1/2}S^T S A_{-i}^{1/2})\\
&= \lambda_{\max}(S A_{\langle -i \rangle} S^T).
\end{align*}

It is clear that if Assumption~\ref{assump:subspace_embedding} is satisfied for $A$ then it is satisfied for $A_{\langle -i \rangle}$ as well.  Thus the previous lemma applies and gives
\[
\lambda_i(SAS^T) \leq \lambda_i(A) + \outlyingthresh \log\frac{\lambda_i(A)}{\outlyingthresh},
\]
as desired.

\end{proof}

\subsection{Two-sided eigenvalue bound.}
By simply combining our upper and lower eigenvalue bounds, we obtain a two-sided bound that we will use throughout.
\begin{theorem}
\label{thm:two_sided_eigenvalue_bound}
Suppose that $S$ is such that Assumption~\ref{assump:subspace_embedding} holds for $A$ where all nonzero eigenvalues of $A$ are at least $\outlyingthresh$ in magnitude.  Then for all $i$ with $\lambda_i(A)\neq 0$ we have 
\[
\abs{\lambda_i(SAS^T) - \lambda_i(A)}
\leq \outlyingthresh \log \frac{\norm{A}{}}{\outlyingthresh}.
\]
\end{theorem}
\begin{proof}
Lemma~\ref{lem:eigval_upper_bounds} and Lemma~\ref{lem:eigval_lower_bound} yield the desired bound for all $i$ with $\lambda_i(A) > 0.$  To obtain the analogous bound for the negative eigenvalues, simply apply these Lemmas to $-A.$
\end{proof}

\section{Bounded Entry Matrices}

As in \cite{bhattacharjee2024sublinear}, we decompose $A$ corresponding to its ``middle" and ``outlying" eigenvalues.
\begin{Definition}
    \label{def:middle_and_outlying}
    Let $A$ be a symmetric matrix with spectral decomposition $A = U \Lambda U^T.$ Let $\outlyingthresh$ be the parameter in Assumption~\ref{assump:subspace_embedding}.  Set $\Lambda_o$ to be $\Lambda$ with all entries smaller than $\outlyingthresh$ zeroed out, and set $\Lambda_m = \Lambda - \Lambda_o$.  Then we define $A_o = U \Lambda_o U^T$ and $A_m = U \Lambda_m U^T$.
\end{Definition}

Throughout we will slightly abuse notation by writing $A_{o,i}$ to refer to row $i$ of the matrix $A_o$, and similarly $A_{o,i,j}$ to refer to the $i,j$ entry of $A_o.$ We use the following bound on the middle eigenvalues from \cite{bhattacharjee2024sublinear}.
\begin{lemma}
\label{lem:lemma_4_of_bha}
Let $A \in \R^{n\times n}$ be symmetric with entries bounded by $1$, and let $A_m$ be $A$ restricted to its middle eigenvalues as defined in Definition~\ref{def:middle_and_outlying}.  Let $S$ be a uniform sampling matrix with sampling probability at least $\frac{c\log n}{\eps^2 \delta}.$ Then with probability at least $1-\delta,$ $\norm{S A_m S^T}{2} \leq \eps n.$
\end{lemma}
\begin{proof}
This is a rephrasing of Lemma 4 of \cite{bhattacharjee2024sublinear}.  Note that our sampling matrix $S$ already rescales by $n/s$, whereas the matrix in their theorem statement is prior to rescaling.
\end{proof}

\begin{theorem}
\label{thm:main_uniform_sampling_result}
Let $A\in \R^{n\times n}$ be symmetric, not necessarily PSD with all entries bounded by $1.$  Then when $s\geq \frac{c}{\eps^2}\frac{\log n}{\delta})$, observing $S A S^T$ allows recovery of all eigenvalues of $A$ to within $\eps n$ additive error with probability at least $1-\delta.$
\end{theorem}
\begin{proof}
Following \cite{bhattacharjee2024sublinear}, write $A = A_o + A_m$ where $A_o$ has the eigenvalues of $A$ that are at least $\eps n$ in magnitude.  For our sampling matrix $S$, Theorem~\ref{thm:two_sided_eigenvalue_bound} implies that the positive eigenvalues of $SA_o S^T$ are $\eps n$ additive approximations to the eigenvalues of $A_o$. Also Lemma~\ref{lem:lemma_4_of_bha} shows that $\norm{SA_mS^T}{} \leq O(\eps n)$ with $0.9$ probability.  So the result follows from Weyl's inequality.
 \end{proof}

\subsection{Removing $\log n$ factors.}  We show via a surprisingly simple trick that the $\log n$ dependence can be replaced with a $\log\frac{1}{\eps \delta}$ dependence, resolving a question left open by \cite{bhattacharjee2024sublinear}. The idea is that we can directly apply our sampling result to get good eigenvalue approximations by sampling a principal submatrix of dimension $\tilde{O}(\log n/\eps^2).$  The new matrix still has a dimension depending on $n$, but is dramatically smaller.  Thus we may apply our sampling procedure again to reduce the dimensions even further.  It turns out that we can repeat this enough times to remove the $n$ dependence entirely.  Since the sampling procedure at each stage is uniform, the final sample is uniform as well.  So this argument does not yield a new algorithm, but rather shows that the $O(\log n /\eps)$ guarantee can be boostrapped to achieve a tighter sampling bound.

\begin{theorem}
Algorithm~\ref{alg:uniform_sampling} approximates all eigenvalues of $A$ to within $\eps n$ additive error, with probability at least $1-\delta$ when $s = \frac{c}{\eps^2\delta}\log\frac{1}{\eps\delta}.$

\end{theorem}

\begin{proof}

Say that a sampling algorithm has a $(\eps, \delta)$ recovery guarantee, if it approximates all eigenvalues of $A$ to within $\eps n$ additive error with failure probability at most $\delta.$  Fix $n$, and let $N\geq 1$ be minimal such that for all $\eps,\delta$ bounded by a sufficiently small constant, setting $s = N/(\eps^2 \delta)$ in Algorithm~\ref{alg:uniform_sampling} suffices for obtaining an $(\eps, \delta)$ approximation guarantee.  (We already know $N$ can be taken to be $c\log n$ by Theorem~\ref{thm:main_uniform_sampling_result}.)

By our definition of $N$, taking $s = 12N/(\eps^2 \delta)$ suffices to obtain an $(\eps/2, \delta/3)$ guarantee.  Let $S$ be the corresponding sampling matrix.  Note that $S$ takes $s$ samples of expectation, and is unlikely to take many more.  By a Chernoff bound, the probability that $S$ takes at most $2s$ samples is at least $1 - \exp(-s/3).$   Note that $s\geq 1/\delta,$ so $\exp(-s/3) \leq \delta$, when $\delta \leq 1/2.$ Also, note that $SAS^T$ has entries bounded by $n/s$ since the entries of $A$ are assumed to be bounded by $1.$

Now we consider sub-sampling again by a new sampling matrix $S_2$. By Theorem~\ref{thm:main_uniform_sampling_result}, if $S_2$ is a uniform sampling matrix that takes $\frac{c}{\delta\eps^2}\log(s/\eps)$ samples in expectation, then with failure probability at most $\delta/3$, the matrix $S_2 S A S^T S_2^T$ yields an additive 
\[
(\eps/4) (2s) \norm{SAS^T}{\infty} = (\eps/2) s (n/s) = \eps n /2
\] approximation to the spectrum of $S A S^T.$ Combining the bounds, we see that $S_2 S A S^T S_2^T$ gives an additive $\eps n$ spectral approximation to $A$ with failure probability at most $\delta.$

Now, note that $S_2 S$ is itself a sampling matrix that takes 
$\frac{c}{\delta\eps^2}\log(s/\eps)$
samples in expectation and yields an $(\eps,\delta)$ guarantee.  By minimality of $N$, this means that
\[
N \leq c\log\frac{s}{\eps}.
\]
But $s = \frac{N}{\eps^2\delta},$ so in fact
\[
N \leq c\log\frac{N}{\eps^3 \delta},
\]
which implies that $N \leq c\log\frac{1}{\eps \delta}.$  Thus we conclude that a sampling matrix with $s = \frac{c}{\eps^2 \delta} \log \frac{1}{\eps \delta}$ is sufficient to obtain an $(\eps, \delta)$ guarantee.
\end{proof}

\section{Squared row-norm Sampling}
In this section we show how to modify the anlaysis of \cite{bhattacharjee2024sublinear} to obtain improved bounds given access to a squared row-norm sampler.

We use the entry zeroing procedure of \cite{bhattacharjee2024sublinear} (the same procedure used in Algorithm~\ref{alg:row_norm_sampling}), which we now recall.
\begin{Definition}
Given a symmetric matrix $A \in \R^{n\times n}$ let $A'$ be the matrix formed by zeroing out all entries $A_{ij}$ satisfying one of the following conditions.
\begin{enumerate}
\item $i = j$ and $\norm{A_i}{}^2 < \frac{\eps^2}{4}\norm{A}{F}^2$
\item $i \neq j$ and $\norm{A_i}{}^2\norm{A_j}{}^2 < \frac{\eps^2\norm{A}{F}^2 \abs{A_{ij}}^2}{c\log^4 n}$ for an absolute constant $c.$
\end{enumerate}
\end{Definition}
We also use a lemma from \cite{bhattacharjee2024sublinear} arguing that the spectrum of $A'$ is close to the spectrum of $A.$
\begin{lemma}
\label{lem:A_prime_and_A_are_close}
For all $i$, $\abs{\lambda_i(A') - \lambda_i(A)} \leq \eps \norm{A}{F}.$
\end{lemma}

We apply the following bound on the operator norm of the sampled matrix, which is given in the proof of Lemma 14 in \cite{bhattacharjee2024sublinear}.
\begin{lemma}
\label{lem:middle_eigenvalue_tropp_bound}
For $s \geq \frac{\log n}{\eps^2},$ we have the bound
\[
\E_2\norm{S A'_m S}{2} \leq 10 \sqrt{\log n}\E_2 \norm{S H_m \hat{S}}{1\rightarrow 2} + 15\eps\norm{A}{F},
\]
where $H_m$ consists of the off-diagonal entries of $A_m'$, $\E_2(X) = \E(X^2)^{1/2}$ and where $\norm{\cdot}{1\rightarrow 2}$ denotes the maximum $\ell_2$ norm of any column.
\end{lemma}

Our main task is to improve the bound on $\E_2 \norm{S H_m \hat{S}}{1\rightarrow 2}$.  To do this, we continue following the proof of \cite{bhattacharjee2024sublinear}, but improve on their variance calculation.

\begin{lemma}
\label{lem:preliminary_variance_bound}
(From \cite{bhattacharjee2024sublinear}) Assume that $p_i \leq 1$ for all $i$. For a fixed $i$, define the random variable $z_j$ to be $\frac{1}{p_j}\abs{A'_{m,i,j}}^2$ with probability $p_j$, and $0$ otherwise.  Then
\[
\Var\left(\sum_{j=1}^n z_j\right)
\leq \sum_{j=1}^n \abs{A'_{m,i,j}}^4 + \sum_{j=1}^n \frac{12\norm{A}{F}^2}{s\norm{A_j}{2}^2}(\abs{A'_{i,j}}^4 + \abs{A'_{o,i,j}}^4)\\
\]
\end{lemma}

\begin{lemma}
Suppose that $s\geq c\frac{\log^4 n}{\eps^2}$ for some absolute constant $c$. Fix $i$ with $p_i < 1$ and let $z_j$ be as in the above lemma. We have
\[
\Var\left(\sum_{j=1}^n z_j\right) \leq 2\norm{A_i}{2}^4.
\]
\end{lemma}
\begin{proof}

From Lemma~\ref{lem:preliminary_variance_bound},
\begin{align*}
\Var\left(\sum_{j=1}^n z_j\right)
&\leq \sum_{j=1}^n \abs{A'_{m,i,j}}^4 + \sum_{j=1}^n \frac{12\norm{A}{F}^2}{s\norm{A_j}{2}^2}(\abs{A'_{i,j}}^4 + \abs{A'_{o,i,j}}^4)\\
&= \sum_{j=1}^n \abs{A'_{m,i,j}}^4 + 12\sum_{j=1}^n \frac{\norm{A}{F}^2}{s\norm{A_j}{2}^2}\abs{A'_{i,j}}^4 + 12\sum_{j=1}^n \frac{\norm{A}{F}^2}{s \norm{A_j}{2}^2}\abs{A'_{o,i,j}}^4.
\end{align*}

Also from the argument of \cite{bhattacharjee2024sublinear}, we have the bound
\[
\abs{A'_{o,i,j}} \leq \frac{\norm{A_i}{2}\norm{A_j}{2}}{\eps\norm{A}{F}}.
\]
(Note that we take their $\delta$ to be $1$.)

Compared to their argument, our main improvement is that we use the bound 
\[
\abs{A'_{o,i,j}}^4 \leq \frac{\norm{A_i}{2}^2 \norm{A_j}{2}^2}{\eps^2 \norm{A}{F}^2}\abs{A'_{o,i,j}}^2.
\]

This gives
\begin{align*}
\sum_{j=1}^n \frac{\norm{A}{F}^2}{s \norm{A_j}{2}^2}\abs{A'_{o,i,j}}^4
&\leq \sum_{j=1}^n \frac{\norm{A}{F}^2}{s \norm{A_j}{2}^2}\left(\frac{\norm{A_i}{2}^2 \norm{A_j}{2}^2}{\eps^2  \norm{A}{F}^2}\abs{A'_{o,i,j}}^2\right)\\
&= \sum_{j=1}^n \frac{\norm{A_i}{2}^2}{s \eps^2 } \abs{A'_{o,i,j}}^2\\
&= \frac{\norm{A_i}{2}^2}{s\eps^2} \norm{A'_{o,i}}{2}^2\\
&\leq \frac{\norm{A_i}{2}^4}{s\eps^2},
\end{align*}
where in the last line, we used that $\norm{A'_{o,i}}{2}^2 \leq \norm{A'_i}{2}^2 \leq \norm{A_i}{2}^2.$

Again following the proof of \cite{bhattacharjee2024sublinear}, we have
\[
\sum_{j=1}^n \abs{A'_{m,i,j}}^4
\leq \norm{A'_{m,i}}{2}^4
\leq \norm{A'_i}{2}^4
\leq \norm{A_i}{2}^4.
\]

Continuing their proof, the thresholding procedure for $A'$ implies that if $i\neq j$ and $A'_{ij}\neq 0$, then
\[
\frac{\norm{A_j}{2}^2}{\abs{A'_{i,j}}^2}
\geq \frac{\eps^2\norm{A}{F}^2}{c\log^4 n \norm{A_i}{2}^2}.
\]
Therefore,
\[
\sum_{j=1}^n \frac{\norm{A}{F}^2}{s\norm{A_j}{2}^2}\abs{A'_{i,j}}^4
= \sum_{j:A'_{ij}\neq 0} \frac{\norm{A}{F}^2}{s\norm{A_j}{2}^2}\abs{A'_{i,j}}^4
\leq \sum_{j=1}^n \frac{c\log^4 n \norm{A_i}{2}^2}{s\eps^2}\abs{A'_{ij}}^2
\leq \frac{c\log^4 n \norm{A_i}{2}^4}{s\eps^2}.
\]
Combining the various bounds gives
\[
\Var\left(\sum_{j=1}^n z_j\right)
\leq \norm{A_i}{2}^4 + \frac{c\log^4 n\norm{A_i}{2}^4}{s\eps^2}.
\]
Thus for $s\geq c\frac{\log^4 n}{\eps^2}$ we have
\[
\Var\left(\sum_{j=1}^n z_j\right) \leq 2\norm{A_i}{2}^4,
\]
as desired.

\end{proof}

With this improved variance bound, we obtain our improved middle eigenvalue bound by continuing the proof of Lemma 14 from \cite{bhattacharjee2024sublinear}.

\begin{lemma}
\label{lem:frobenius_middle_eigenvalue_bound}
We have the operator norm bound $\norm{S A_m' S^T}{} \leq \eps \norm{A}{F}$ with $2/3$ probability provided that $s \geq c\frac{\log^4 n}{\eps^2}$ and $p_i := \frac{s\norm{A_i}{2}}{\norm{A}{F}} \leq 1$ for all $i.$
\end{lemma}
\begin{proof}
In their Lemma 14, \cite{bhattacharjee2024sublinear} shows that $\abs{z_j} \leq 2\norm{A_i}{2}^2$ for $s\geq c\frac{\log^4 n}{\eps^2}$ if $p_i < 1.$  Using our improved variance bound above, and applying Bernstein's inequality as in \cite{bhattacharjee2024sublinear} gives
\[
\Pr(\norm{(SA_m')_{:,i}}{2}^2 \geq \E\norm{(SA'_m)_{:,i}}{2}^2 + t)
\leq \left( \sum_{j=1}^n z_j \geq \norm{A_i}{}^2 + t\right)
\leq \exp\left(c\frac{-t^2}{\norm{A_i}{2}^4 + t\norm{A_i}{2}^2}\right).
\]
Then setting $t = c\log n\norm{A_i}{2}^2$ gives
\[
\Pr(\norm{(SA_m')_{:,i}}{2}^2 \geq \E\norm{(SA'_m)_{:,i}}{2}^2 + c\log n\norm{A_i}{2}^2)
\leq 1/n^4,
\]
for appropriate constants. It then follows that
\[
\frac{1}{p_i}\norm{(SH_m)_{:,i}}{}^2 \leq \frac{1}{p_i}\norm{(SA'_m)_{:,i}}{}^2
\leq \frac{1}{p_i}(c\log n \norm{A_i}{2}^2)
\leq \frac{\eps^2 \norm{A}{F}^2}{\log n},
\]
for $s\geq c\frac{\log^4 n}{\eps^2}$ with failure probability at most $1/n^4.$  Thus with failure probability at most $1/n^3$ this bound holds for all $i$ simultaneously.

In order to apply this, we would instead like a second moment bound.  Let $E$ be the event that all $p_i$ corresponding to sampled rows are bounded by $1/n^2.$  Note that $E$ occurs with probability at least $1-1/n.$ From now on we condition on $E$ with only a $1/n$ loss in probability.  Then we have
\[
\frac{1}{\sqrt{p_i}}\norm{(SH_m)_{:,i}}{} \leq n\norm{A}{F}
\]
deterministically, and our bound above then implies that $\E_2 \norm{S H_m \hat{S}}{1\rightarrow 2}$ is bounded by $\eps \norm{A}{F}.$

Finally plugging into Lemma~\ref{lem:middle_eigenvalue_tropp_bound} we have
\[
\E_2 \norm{SA'_m S}{2} \leq c \eps \norm{A}{F},
\]
and the desired bound follows from Markov's inequality.
\end{proof}

This is sufficient to show correctness of the ``restricted" sampling algorithm, Algorithm~\ref{alg:restricted_row_norm_sampling}.

\input{restricted_row_norm_sampling}

\begin{lemma}
\label{lem:restricted_sampling_correctness}
Suppose that $s\geq c\frac{\log^4 n}{\eps^2}\log^2\frac{1}{\eps}$ and that $A$ and $s$ satisfy the condition of Algorithm~\ref{alg:restricted_row_norm_sampling} that the row sampling probabilities $p_i$ are bounded by $1.$  Then with $2/3$ probability, Algorithm~\ref{alg:restricted_row_norm_sampling} returns an additive $\eps\norm{A}{F}$ approximation the spectrum of $A.$
\end{lemma}

\begin{proof}
As discussed above, write 
\[ 
S A' S^T = S A'_o S^T + SA'_m S^T.
\]
By Lemma~\ref{lem:subspace_embedding_from_row_norm_sampling}, Assumption~\ref{assump:subspace_embedding} holds for $A'_o$ and $S$ with $\outlyingthresh=\frac{\eps}{\log\frac{1}{\eps}}\norm{A}{F}.$  As a result, by Theorem~\ref{thm:two_sided_eigenvalue_bound},  the eigenvalues of $SA'_o S^T$ are an additive $c\eps\norm{A}{F}$ approximation to the eigenvalues of $A'_o$, which are in turn additive $\eps\norm{A}{F}$ approximations to the eigenvalues of $A_o$ by Lemma~\ref{lem:A_prime_and_A_are_close}.

Also by Lemma~\ref{lem:frobenius_middle_eigenvalue_bound} we have $\norm{SA'_m S^T}{} \leq c\eps\norm{A}{F}$, so correctness follows from Weyl's inequality \cite{weyl1912asymptotic}.
\end{proof}

Finally it remains to deal with the situation where the restricted condition of Algorithm~\ref{alg:restricted_row_norm_sampling} is not met.  The idea is simple -- for each row with $p_i$ larger than $1$, we simple split that row into several scaled-down rows.  This will not affect the spectrum, but will allow the condition of Algorithm~\ref{alg:restricted_row_norm_sampling} to be met.  In fact, for simplicity we split all rows in this way although this is not strictly necessary.

\begin{theorem}
\label{thm:main_row_norm_sampling_result}
Algorithm~\ref{alg:row_norm_sampling} yields an $\eps \norm{A}{F}$ additive approximation to the spectrum of $A$ with $2/3$ probability when $s\geq \frac{c\log^4 n}{\eps^2}\log^2\frac{1}{\eps}.$
\end{theorem}
\begin{proof}
We consider an algorithm which will be equivalent to Algorithm~\ref{alg:row_norm_sampling}.
Set $p_i = \frac{s\norm{A_i}{2}^2}{\norm{A}{F}^2}.$  Note that $p_i\leq s$ for all $i$, so it will suffice to ``split" each row into roughly $s$ smaller rows.  More formally, let $U$ be a vertical stack of consisting of $s$ copies of $\frac{1}{\sqrt{s}} I_n.$  Note that $U$ has orthonormal columns and that the image of $U^T$ is $\R^n.$   This means that the nonzero spectrum of $U A U^T$ coincides with the nonzero spectrum of $A.$  In particular we also have $\norm{A}{F} = \norm{UAU^T}{F}.$ Note that the norm of row $i$ of $UAU^T$ is $1/\sqrt{s^2}$ times the norm of a corresponding row $i'$ of $A.$

Thus for all $i,$
\[
\frac{s\norm{(UAU^T)_i}{2}^2}{\norm{UAU^T}{F}^2}
\leq \frac{1}{s}\frac{s\norm{A_{i'}}{2}^2}{\norm{A}{F}^2}
= \frac{\norm{A_{i'}}{2}^2}{\norm{A}{F}^2}
\leq 1.
\]

Thus Lemma~\ref{lem:restricted_sampling_correctness} establishes correctness of the sampling procedure for $UAU^T$ (which has the same nonzero spectrum of $A$).  In order to achieve $\eps\norm{UAU^T}{F} = \eps\norm{A}{F}$ additive error, we need
$s \geq \frac{c\log^4(ns^2)}{\eps^2},$ which is achieved for $s\geq \frac{c\log^4 n}{\eps^2}$ (possibly after adjusting the constant $c$).

Now to simulate Algorithm~\ref{alg:restricted_row_norm_sampling} on $UAU^T$, note that we can simply sample $\text{Binomial}(s, \frac{\norm{A_i}{2}^2}{\norm{A}{F}^2})$ copies of each index $i$, as is done in Algorithm~\ref{alg:row_norm_sampling}.

We finally note that for $s\geq c/\eps^4$ that the diagonal zeroing condition for $UAU^T$ in Algorithm~\ref{alg:row_norm_sampling} is always satisfied.  Moreover the condition for zeroing out off-diagonal entries is unchanged since applying $U$ to $A$ scales row norms by $1/\sqrt{s}$ but entries by $1/s.$
\end{proof}

\section{Additional Applications}

\subsection{Application to Sketching}

\cite{swartworth2023optimal} showed that a sketching dimension of $\frac{1}{\eps^2}$ is sufficient to approximate all eigenvalues of $A$ to within $\eps \norm{A}{F}$ additive error.  Unfortunately, the sketch consisted of a dense Gaussian matrix, which means that it requires roughly $n^2/\eps^2$ time to apply to a dense matrix $A$\footnote{Some speedup is possible by using fast matrix multiplication, although these algorithms are impractical for reasonably sized inputs.}.  As a consequence of our sampling results, we show that it is possible to match the optimal sketching dimension with a sketch that can be applied in time linear in the number of entries of $A.$  We observe that conjugating by a Hadamard matrix is sufficient to flatten the row norms, so that our row-norm sampling procedure reduces to uniform sampling.  The idea of applying Hadamard matrices to improve the runtimes of sketches is well-known in the literature (see for example \cite{woodruff2014sketching,tropp2011improved}).  Importantly, if $H$ is Hadamard,  then a matrix-vector product $Hv$ can be carried out in $O(n\log n)$ time. Similarly if $A\in\R^{n\times n}$ is a matrix, then $HA$ can be carried out in $O(n^2\log n)$ by applying $H$ columnwise.

\begin{theorem}
There is a bilinear sketch $\R^{n\times n} \rightarrow \R^{k\times k}$ with $k=O(1/\eps^2)$ that can be applied in $O(n^2)$ time, and allows all eigenvalues of $A$ to be recovered to within $\eps \norm{A}{F}$ additive error.
\end{theorem}

\begin{proof}
Let $U = HD \in \R^{n\times n}$ where $D$ has i.i.d. random signs on its diagonal and $H$ is a Hadamard matrix scaled by $1/\sqrt{n}$. Note that $U$ is an orthogonal matrix, and so
\[
\norm{(U A U^T)_i}{}^2
= \norm{U A U^T e_i}{}^2
= (U^T e_i)^T A^2 (U^T e_i).
\]
Also note that $U^T e_i = D H e_i$ which is distributed as a  Rademacher random vector scaled by $1/\sqrt{n}$.  Recalling that $\E(x^T A^2 x) = \tr(A^2) = \norm{A}{F}^2$, when $x$ is a random sign vector, by Hanson-Wright we have the bound
\[
\Pr(\abs{(UAU^T)_i - \frac{1}{n}\norm{A}{F}^2}\geq t) \leq 2\exp(-c)
\leq 2\exp\left(-c \min\left(\frac{nt}{\norm{A^2}{F}}, (\frac{nt}{\norm{A^2}{F}})^2\right)\right).
\]
Setting $t \geq c\frac{\norm{A}{F}^2}{n} \log\frac{n}{\delta}$ makes the the right hand side bounded by $\frac{\delta}{n},$ since $\norm{A^2}{F} \leq \tr(A^2) = \norm{A}{F}^2.$  By a union bound we have $\abs{(UAU^T)_i}{}^2 \leq c\frac{\norm{A}{F}^2}{n}\log\frac{n}{\delta}$ for all $i$, with probability at least $1-\delta.$

It follows from Theorem~\ref{thm:main_row_norm_sampling_result} that it suffices to run Algorithm~\ref{alg:row_norm_sampling} by uniformly sampling each row with probability $O(\frac{1}{n\eps^2}\poly\log\frac{n}{\eps})$, resulting in a submatrix of dimension $k\times k,$ where $k = O(\frac{1}{\eps^2}\poly\log\frac{n}{\eps})$ with high probability. To implement this as a sketch, we can choose a uniformly random permutation matrix $P$, (implicitly) form $(PHD)A(PHD)^T$ and then sample a leading principal minor of dimension $O(\frac{1}{\eps^2}\poly\log\frac{n}{\eps})$. 

Finally, in parallel we can sketch the norms of the first $1/\eps^2$ columns of $(PHD)A(PHD)^T$ to within a constant factor using a Johnson-Lindenstrauss sketch (see \cite{woodruff2014sketching} for example).  This only takes an additional $\frac{1}{\eps^2}\log\frac{1}{\eps\delta}$ space to succeed with $1-\delta$ probability on all of the first $1/\eps^2$ rows, and can applied in $O(\frac{n}{\eps^2}\poly\log\frac{n}{\eps})$ time.  This allows us to run Algorithm~\ref{alg:row_norm_sampling} on the resulting matrix.
\end{proof}

\subsection{Improvements by combining with adaptive results}
\input{combining_with_adaptive}

\section{Top Eigenvector Estimation}

In this section, we analyze Algorithm~\ref{alg:top_eigenvector} for producing an approximate top eigenvector for PSD $A$ with $\norm{A}{\infty}\leq 1.$ To prove correctness, we begin with a few simple facts that we will need below.

\begin{lemma}
\label{lem:psd_incoherence}
Suppose that $A\in\R^{n\times n}$ is PSD with entries bounded by $1$, and let $v$ be an eigenvector of $A$ with associated eigenvalue $\lambda$. Then $\norm{v}{\infty} \leq \frac{1}{\sqrt{\lambda}}$
\end{lemma}
\begin{proof}
Write the spectral decomposition of $A$ as
\[
A = \sum_{i=1}^n \lambda_i v_i v_i^T,
\]
where each $v_i$ is an eigenvector of $A$ with associated eigenvalue $\lambda_i.$  For arbitrary $j$, we then have
\[
A_{jj} = \sum_{i=1}^n \lambda_i v_{ij}^2.
\]
By the bounded entry hypothesis, this is bounded by $1$, so in particular $\lambda_1 v_{1j}^2 \leq 1$ since $A$ is PSD.  Since $j$ was arbitrary, the claim follows.
\end{proof}

A related fact is that zeroing out eigenvalues of $A$ only decreases the diagonal entries.  We only need this fact for the top eigenvalue, but it is true more generally.
\begin{lemma}
\label{lem:diag_entries_decrease}
Let $A_{-1}$ denote $A$ with its top eigenvalue zeroed out. Then $(A_{-1})_{kk} \leq A_{kk}$ for all $k$
\end{lemma}
\begin{proof}
Write the spectral decomposition of $A$ as
\[
A = \sum_{i=1}^n \lambda_i v_i v_i^T.
\]
Now simply note that 
\[
\lambda_1 (v_1 v_1^T)_{kk} = \lambda_1 v_{1k}^2 \geq 0.
\]
\end{proof}

The following fact states that applying a half iteration of power method only increases the Rayleigh quotient. (This statement is true more general, but we only need the $1/2$ version.)
\begin{lemma}
\label{lem:power_method_only_helps}
Let $A\in\R^{n\times n}$ be PSD.  Then for all nonzero $x\in\R^n$,
\[
\frac{x^T A^2 x}{x^T A x} \geq \frac{x^T A x}{x^T x}.
\]
\end{lemma}

\begin{proof}

The desired inequality is equivalent to
\[
\inner{Ax}{x}^2 \leq \norm{Ax}{2}^2 \norm{x}{2}^2,
\]
which is true by Cauchy-Schwarz.

\end{proof}

The following is our main lemma for top eigenvector approximation.
\begin{lemma}
\label{lem:main_eigvect_lemma}
Let $S$ be a column-sampling matrix that samples each column of $A$ independently with probability $p=\frac{c}{\eps n}$ for an absolute constant $c$.  Then with probability at least $2/3$, there is some $x$ in the image of $S,$ such that
\[
\frac{x^T A^2 x}{x^T A x} \geq \lambda_1 - \eps n.
\]
\end{lemma}

\begin{proof}

If $\lambda_1 < \eps n$ then the result is trivial.  So we assume from now on that $\lambda_1 \geq \eps n.$

First, for an arbitrary fixed $x\in \R^{n}$, consider the generalized Rayleigh quotient $\frac{x^T A^2 x}{x^T A x}.$  Let $v^{(1)}, \ldots, v^{(n)}$ be an orthonormal basis of eigenvector of $A$ with associated eigenvalues $\lambda_1\geq \lambda_2 \geq \ldots \geq \lambda_n.$  Then we have
\begin{align*}
\frac{x^T A^2 x}{x^T A x} 
&= \frac{\lambda_1^2\inner{x}{v^{(1)}}^2  + \lambda_2^2 \inner{x}{v^{(2)}}^2  + \ldots + \lambda_n^2 \inner{x}{v^{(n)}}^2 }{ \lambda_1 \inner{x}{v^{(1)}}^2 + \lambda_2 \inner{x}{v^{(2)}}^2 + \ldots + \lambda_n \inner{x}{v^{(n)}}^2}\\
&\geq \frac{\lambda_1^2\inner{x}{v^{(1)}}^2}{ \lambda_1 \inner{x}{v^{(1)}}^2 + \lambda_2 \inner{x}{v^{(2)}}^2 + \ldots + \lambda_n \inner{x}{v^{(n)}}^2}\\
&= \frac{\lambda_1^2}{ \lambda_1 + \frac{1}{\inner{x}{v^{(1)}}^2}(\lambda_2 \inner{x}{v^{(2)}}^2 + \ldots + \lambda_n \inner{x}{v^{(n)}}^2)}\\
&\geq \lambda_1 - \frac{1}{\inner{x}{v^{(1)}}^2}(\lambda_2 \inner{x}{v^{(2)}}^2 + \ldots + \lambda_n \inner{x}{v^{(n)}}^2),
\end{align*}
where the last inequality follows from the difference-of-squares factorization.  Now write $A_{-1}$ to denote $A$ with its top eigenvalue zeroed out. The term in parentheses above is $x^T A_{-1} x$, so we have
\[
\frac{x^T A^2 x}{x^T A x} 
\geq \lambda_1 - \frac{x^T A_{-1} x}{\inner{x}{v}^2}.
\]
It therefore suffices to show that $\frac{x^T A_{-1} x}{\inner{x}{v}^2} \leq \eps n$ for some $x$ in the image of $S.$  We will take $x = \Pi_S v$ where $v = v^{(1)}$ is the top eigenvector of $A$ and where $\Pi$ is the orthogonal projection onto the image of $S.$  We assume that $v$ is normalized to have unit norm.  For this choice of $x$ we will bound the numerator and denominator separately.

\paragraph{Denominator bound.}  First note that 
\[
\inner{x}{v}^2 = \inner{\Pi v}{v}^2 = \inner{\Pi v}{\Pi v}^2 = \norm{\Pi v}{2}^4,
\]
so we simply need to bound $\norm{\Pi v}{2}^2.$ We will prove the denominator bound under the assumption that $p=\frac{10}{\eps n}.$ Clearly this is sufficient as the denominator for larger $p$ majorizes the denominator for smaller $p.$  By Lemma~\ref{lem:psd_incoherence}, recall that $\abs{v_i}^2 \leq \frac{1}{\lambda_1} \leq \frac{1}{\eps n}$ for all $i.$  Let $\sigma_i$ be i.i.d. $\text{Bernoulli}(p)$ random variables.  We wish to obtain a lower bound on
$X := \sum_i \sigma_i v_i^2.$ Since $v$ is a unit vector, by linearity of expectation $\E X = p.$  For the second moment,
\begin{align*}
\E X^2 &= p(v_1^4 + \ldots + v_n^4) + 2p^2 \sum_{i < j} v_i^2 v_j^2\\
&= p(v_1^4 + \ldots + v_n^4) + p^2 (v_1^2 + \ldots + v_n^2)^2 - p^2(v_1^4 + \ldots + v_n^4)\\
&\leq p(v_1^4 + \ldots + v_n^4) + p^2\\
&\leq \frac{p}{\eps n} (v_1^2 + \ldots + v_n^2) + p^2\\
&= \frac{p}{\eps n} + p^2.
\end{align*}


Therefore $\Var(X) \leq \frac{p}{\eps n},$ and so by Chebyshev's inequality,
\[
\Pr(X \leq \frac{1}{\eps n})
\leq \Pr(\abs{X - p} \geq \frac{9}{\eps n})
\leq \frac{\Var(X)}{(9/(\eps n))^2}
\leq \frac{\eps n p}{81} = \frac{10}{81}.
\]
Thus we have $\inner{x}{v}^2 \geq \frac{1}{(\eps n)^2}$ with probability at least $7/8.$

\paragraph{Numerator bound.}
Since $A_{-1}$ is PSD, we can write $A_{-1} = U^T U$ for some $U\in\R^{n\times n}$.  By Lemma~\ref{lem:diag_entries_decrease}, the diagonal entries of $A_{-1}$ are all bounded by 1, which means that the columns of $U$ each have $\ell_2$ norm at most $1$. Also, since $v$ is the top eigenvector of $A$, we have $v^T A_{-1} v = 0$, which implies that $Uv = 0.$

Let $w^{(i)} = v_i U_i$ where $v_i$ is the $i$th entry of $v$ and $U_i$ is the $i$th column of $U$. Also set $p=\frac{1}{\eps n}.$

Note that we have 
\[
\sum_{i=1}^n w^{(i)} = 0,
\] 
and 
\[
\sum_{i=1}^n \norm{w^{(i)}}{2}^2 \leq 1.
\]
We wish to bound
\[
\sum_{i=1}^n \norm{\sigma_i w^{(i)}}{}^2,
\]
where the $\sigma_i$ are i.i.d. $\text{Bernoulli}(p)$.
Since the quantity we wish to bound is non-negative, it suffices to bound its expectation:
\begin{align*}
\E \norm{\sum_i \sigma_i w^{(i)}}{2}^2 
&= p \sum_{i=1}^n ||w^{(i)}||^2 + 2p^2 \sum_{i < j} \inner{w^{(i)}}{w^{(j)}}\\
&\leq p + 2p^2 \sum_{i < j} \inner{w^{(i)}}{w^{(j)}}.
\end{align*}

This last sum can be rewritten as
\[
2\sum_{i < j} \inner{w^{(i)}}{w^{(j)}}
= \norm{w^{(1)} + \ldots + w^{(n)}}{2}^2 - \left(\norm{w^{(i)}}{2}^2 + \ldots + \norm{w^{(n)}}{2}^2\right)
\leq 0,
\] 
since $w^{(1)} + \ldots + w^{(n)} = 0.$

It follows that 
\[
\E \norm{\sum_i \sigma_i w^{(i)}}{2}^2
\leq p,
\]
so by Markov's inequality
\[
\norm{\sum_i \sigma_i w^{(i)}}{2}^2 \leq 10p,
\]
with probability at least $9/10.$

Combining the bound on the numerator and the denominator shows that $\frac{x^T A_{-1} x}{x^T A x} \geq \lambda_1 - 100\eps n.$  The result follows by replacing $\eps$ with $\eps/100.$

\end{proof}

Finally we prove correctness of Algorithm~\ref{alg:top_eigenvector}.
\begin{theorem}
Let $A\in\R^{n\times n}$ be symmetric PSD with $\norm{A}{\infty}\leq 1$.  For $p = \frac{c}{\eps n}$, with $3/4$ probability, Algorithm~\ref{alg:top_eigenvector} returns a unit vector $u\in\R^n$ satisfying
\[
u^T A u \geq \lambda_1 - \eps n,
\]
where $\lambda_1$ is the top eigenvalue of $A.$
\end{theorem}

\begin{proof}
Consider the $x \in \R^n$ produced in Algorithm~\ref{alg:top_eigenvector} which maximizes $\frac{x^T S^T A^2 S x}{x^T S^T A S x}$.  By Lemma~\ref{lem:main_eigvect_lemma}, with $3/4$ probability we have that 
\[
\frac{x^T S^T A^2 S x}{x^T S^T A S x} 
= \frac{(A^{1/2} S x)^T A (A^{1/2} S x)}{(A^{1/2} S x)^T (A^{1/2} S x)}
\geq \lambda_1 - \eps n.
\]
Now from Lemma~\ref{lem:power_method_only_helps}, an additional half-iteration of power method only helps.  In other words,
\[
\frac{(ASx)^T A (ASx)}{(ASx)^T(ASx)} 
\geq 
\frac{(A^{1/2} S x)^T A (A^{1/2} S x)}{(A^{1/2} S x)^T (A^{1/2} S x)}
\geq \lambda_1 - \eps n,
\]
and the main result follows.
\end{proof}

%% file: restricted_row_norm_sampling.tex
\begin{algorithm}
\caption{Restricted Row-norm Sampling Algorithm}
\label{alg:restricted_row_norm_sampling}
\begin{enumerate}
\item Input: Symmetric matrix $A\in \R^{n\times n}$ along with its row norms, sample size $s$.  The inputs are restricted so that $A$ and $s$ must satisfy $\frac{s\norm{A_i}{}^2}{\norm{A}{}^2} \leq 1$ for all $i.$

\item Let $S\in \R^{k\times n}$ be a (rescaled) sampling matrix which samples each row $i$ of $A$ independently with probability $\frac{s\norm{A_i}{}^2}{\norm{A}{}^2}.$

\item Implicitly form the matrix $A' \in \R^{n\times n}$ defined by
\[(A')_{ij} = 
\begin{cases} 
      0 & i=j\,\,\text{and}\,\,\norm{A_i}{2}^2 \leq \frac{\eps^2}{4}\norm{A}{F}^2 \\
      0 & i\neq j\,\,\text{and}\,\,\norm{A_i}{2}^2\norm{A_j}{2}^2 \leq \frac{\eps^2 \norm{A}{F}^2\abs{A_{ij}}^2}{c\log^4 n}\\
      A_{ij} & \text{otherwise}
\end{cases}
\]

\item Return the $k$ eigenvalues of $S^T A' S$, along with $n-k$ additional $0$'s.
\end{enumerate}
\end{algorithm}

%% file: combining_with_adaptive.tex
By combining with known adaptive sampling bounds of \cite{musco2017recursive}, we point out that one can approximate the eigenvalues of a bounded entry PSD $A$ to within $\eps n$ additive error using just $\tilde{O}(1/\eps^3)$ entry queries.  We note that the PSD sampling result that we use was already present in \cite{bhattacharjee2024sublinear}, although this improvement using adaptivity was not pointed out.

We borrow a result from \cite{musco2017recursive}, slightly restated for our setting.
\begin{theorem}
(\cite{musco2017recursive} Theorem 3.) 
Let $A\in\R^{n\times n}$ be PSD and $\lambda > 0$.
 There is an algorithm that adaptively queries $O(n d_{\lambda} \log d_{\lambda})$ entries of $A$ and with $2/3$ probability produces a spectral approximation $\hat{A}$ of $A$ satisfying 
\[
\norm{A - \hat{A}}{} \leq \lambda.
\]
Here $d_{\lambda} := \text{tr}(A(A+\lambda I)^{-1})$ is the $\lambda$-effective dimension of $A.$
\end{theorem}

We observe that this result improves sampling guarantees using adaptivity.

\begin{theorem}
\label{thm:musco_recursive_nystrom}
Let $A$ be PSD with $\norm{A}{\infty} \leq 1.$ Then there is an algorithm that adaptively queries $\tilde{O}(1/\eps^3)$ entries of $A$ and with $2/3$ probability produces an additive $\eps n$ approximation to the spectrum of $A.$
\end{theorem}
\begin{proof}

For the first step we apply our Theorem~\ref{thm:main_uniform_sampling_result} (or \cite{bhattacharjee2024sublinear} since $A$ is PSD), to (implicitly) produce a rescaled submatrix $S^T A S$ of dimension $\tilde{O}(1/\eps^2)$ whose eigenvalues contain additive $\eps n$ approximations to all eigenvalues of $A$ that are at least $\eps n.$ Note that we do not actually form the matrix yet, however we have query access to its entries simply by querying entries of $A.$ By our Theorem~\ref{thm:main_row_norm_sampling_result}, outputting the spectrum of $S^T A S$ along with additional zeros would suffice.

To achieve an $O(\eps n)$ additive error overall, it suffices to approximate the eigenvalues of $S^T A S$ to within $O(\eps n)$ additive error.  For this we apply Theorem~\ref{thm:musco_recursive_nystrom} to obtain $\hat{A}$ with $\norm{A - \hat{A}}{} \leq \eps n.$  Note that we have
\[
d_{\lambda} 
= \text{tr}(A(A + \eps n I)^{-1})
\leq \sum_{i=1}^n \frac{\lambda_i}{\eps n}
= \frac{\text{tr}(A)}{\eps n} \leq \frac{1}{\eps}.
\]
Therefore only $\tilde{O}(1/\eps^3)$ queries are needed to produce $\hat{A}.$  Finally, by Weyl's inequality \cite{weyl1912asymptotic}, 
$\abs{\lambda_i(A) - \lambda_i(\hat{A})} \leq \eps n$ for all $i$.  The result follows by replacing $\eps$ with $\eps/c$ for an appropriate constant.
\end{proof}